\documentclass[a4paper,fleqn]{cas-dc}

\usepackage[numbers]{natbib}
\usepackage{ragged2e}
\usepackage{tabularx} 
\usepackage{amsthm}
\usepackage{amsmath}
\newtheorem{definition}{Definition}
\newtheorem{theorem}{Theorem}

\newtheorem{proposition}{Proposition}
\newtheorem{corollary}{Corollary}
\newcommand{\Challenge}{\mathsf{challenge}}
\newcommand{\Response}{\mathsf{response}}
\newcommand{\MasterKDF}{\mathsf{HKDF}_{\mathrm{master}}}
\newcommand{\ChallengeLabel}{\text{\texttt{"challenge"}}}
\newcommand{\ResponseLabel}{\text{\texttt{"response"}}}
\usepackage{algorithm}
\usepackage{algorithmic}
\usepackage{booktabs}
\usepackage{enumitem} 
\usepackage{multirow}
\usepackage{placeins}
\usepackage{graphicx} 
\usepackage[protrusion=true,expansion=true]{microtype}
\usepackage{tikz}
\usetikzlibrary{arrows.meta,calc,positioning,shapes.geometric}
\usepackage{bm}
\usepackage{hyperref}
\usepackage{soul}
\usepackage{mdframed}
\sethlcolor{blue!20}
\newcommand{\rev}[1]{#1}
\hypersetup{
    hypertexnames=false, 
    colorlinks=true,      
    linkcolor=blue,       
    citecolor=blue,       
    urlcolor=blue,        
    allcolors=blue        
}

\ExplSyntaxOn
\cs_set:Npn \__cas_head:
  {
    \parbox[t]{\textwidth}
      {\sffamily\small\centering \__short_title:}
  }
\cs_set:Npn \__first_head:
  {\parbox[t]{\textwidth}{\rule{\textwidth}{0pt}}}
\cs_set:Npn \__cas_foot:
  {
    \parbox[t]{\textwidth}
      {
        \rule{\textwidth}{.2pt}\\
        \sffamily\small
        \__short_authors: :~{\rmfamily\itshape High-Confidence~Computing}
        \hfill Page~\thepage{}~of~\number\numexpr\lastpage-1\relax
      }
  }
\cs_set_eq:NN \__first_foot: \__cas_foot:
\ExplSyntaxOff
\makeatletter
\AddToHook{cmd/printFirstPageNotes/before}{\begingroup\let\@footnotetext\H@@footnotetext}
\AddToHook{cmd/printFirstPageNotes/after}{\endgroup}
\makeatother

\usepackage[none]{hyphenat}

\usepackage{comment}
\setlist{topsep=3pt,itemsep=1pt,parsep=0pt,partopsep=0pt}
\def\tsc#1{\csdef{#1}{\textsc{\lowercase{#1}}\xspace}}
\tsc{WGM}
\tsc{QE}
\tsc{EP}
\tsc{PMS}
\tsc{BEC}
\tsc{DE}

\begin{document}
\let\WriteBookmarks\relax
\def\floatpagepagefraction{1}
\def\textpagefraction{.001}
\shorttitle{HPQ-AKE for Bandwidth-Constrained IoT and Edge Networks}
\ifdefined\BlindManuscript
\shortauthors{Anonymous Author(s)}
\else
\shortauthors{Khiem P.T. et~al.}
\fi

\title [mode = title]{HPQ-AKE: A Provably Secure Sign-Less Hybrid Authenticated Key Exchange Protocol for Bandwidth-Constrained IoT and Edge Networks}                      



\ifdefined\BlindManuscript
\author[]{Anonymous Author(s)}
\else
\author[1,2]{Khiem Pham-Tuan}[
                        orcid={0009-0007-4453-1240}]
\ead{khiempt.25ai@ou.edu.vn}
\ead{khiempt@huit.edu.vn}

\credit{Conceptualization, Methodology, Software, Formal Analysis, Writing - Original Draft}

\affiliation[1]{organization={Ho Chi Minh City Open University},
                addressline={97 Vo Van Tan Street, Xuan Hoa Ward}, 
                city={Ho Chi Minh City},
                country={Viet Nam}}
\affiliation[2]{organization={Faculty of Information Technology, Ho Chi Minh City University of Industry and Trade (HUIT)},
                addressline={140 Le Trong Tan Street, Tay Thanh Ward}, 
                city={Ho Chi Minh City},
                country={Viet Nam}}                
\author[3]{Minh Quang Le}[orcid={0000-0002-4109-8054}]
\cormark[1]
\credit{Conceptualization, Methodology, Writing - Original Draft, Validation, Supervision}
\ead{minh.lq@ou.edu.vn}
\affiliation[3]{organization={Faculty of Information Technology, Ho Chi Minh City Open University},
                addressline={97 Vo Van Tan Street, Xuan Hoa Ward}, 
                city={Ho Chi Minh City},
                country={Viet Nam}}

\author[4,5]{Khuong Nguyen-An}[orcid={0000-0002-9910-6387}]
\cormark[1]
\credit{Conceptualization, Methodology, Formal Analysis, Validation, Writing - Review \& Editing, Supervision, Project Administration}
\ead{nakhuong@hcmut.edu.vn}
\affiliation[4]{organization={Faculty of Computer Science and Engineering, Ho Chi Minh City University of Technology
(HCMUT)},
                addressline={268 Ly Thuong Kiet Street, Dien Hong Ward}, 
                city={Ho Chi Minh},
                country={Viet Nam}}

\affiliation[5]{organization={Vietnam National University Ho Chi Minh City},
                addressline={Linh Xuan Ward}, 
                city={Ho Chi Minh City},
                country={Viet Nam}}
                
\cortext[cor1]{Corresponding authors}
\fi


\begin{abstract}
Securing bandwidth-constrained Internet of Things (IoT) and edge networks during post-quantum migration creates an authentication trade-off: ML-KEM provides post-quantum key establishment, whereas post-quantum signatures increase certificate-chain size, verification cost, and handshake latency. We present HPQ-AKE, a sign-less hybrid authenticated key exchange for RSA-enabled IoT gateways and edge/cloud services that replaces transcript signatures with dual KEMs. HPQ-AKE combines ML-KEM-768 for post-quantum session secrecy and forward secrecy with RSA-OAEP for implicit mutual authentication, preserving existing RSA-enabled infrastructure during migration. We analyze HPQ-AKE in an extended Bellare--Rogaway model, separating classical authentication in the Random Oracle Model from session-key secrecy in the Quantum Random Oracle Model. Under the Module-LWE and RSA assumptions, it establishes AKE security, forward secrecy, and \rev{conditional KCI resistance under the long-term-key-only exposure assumptions.} Evaluation combines x86 micro-benchmarking, analytical latency modeling, and 10,000-iteration Monte Carlo simulation. HPQ-AKE reduces modeled handshake transmission from 13,009 to 5,668 Bytes, a 56.4\% reduction relative to a Hybrid TLS 1.3 Full handshake, while requiring 7.11 ms of total local computation on the measured x86 testbed. Under a simulated 50 kbps satellite-like link with 600 ms RTT and stochastic jitter, median latency is 1,914 ms, 31.4\% below the 2,791 ms baseline. \rev{At 20 ms RTT, modeled break-even thresholds are 0.31 Mbps against the pre-cached baseline and 4.08 Mbps against the full baseline; both decrease as RTT increases.} \rev{These x86-based results motivate evaluation on gateway-class IoT and edge platforms; performance on ARM gateways and unaccelerated microcontrollers remains unvalidated.}
\end{abstract}

\begin{graphicalabstract}
\resizebox{\linewidth}{!}{\definecolor{gaNavy}{HTML}{17364D}
\definecolor{gaBlue}{HTML}{1199C4}
\definecolor{gaCyan}{HTML}{DFF5FA}
\definecolor{gaGreen}{HTML}{349B68}
\definecolor{gaPaleGreen}{HTML}{E7F5EC}
\definecolor{gaOrange}{HTML}{E87522}
\definecolor{gaPaleOrange}{HTML}{FFF0E4}
\definecolor{gaPurple}{HTML}{7650A8}
\definecolor{gaPalePurple}{HTML}{F0E9F8}
\definecolor{gaGray}{HTML}{687984}
\definecolor{gaLightGray}{HTML}{E4EAED}
\definecolor{gaRed}{HTML}{D94C45}

\begin{tikzpicture}[x=1.25cm,y=1.25cm,font=\sffamily,every node/.append style={execute at begin node=\sffamily}]
  \path[use as bounding box] (0,0) rectangle (16.25,6.25);
  \fill[white] (0,0) rectangle (16.25,6.25);

  \node[font=\bfseries\fontsize{11.8}{12.5}\selectfont,text=gaNavy]
    at (6.5,4.73) {HPQ-AKE: Gateway-Assisted Post-Quantum Migration for IoT};
  \draw[gaLightGray,line width=0.7pt] (0.35,4.43)--(12.65,4.43);

  \node[font=\bfseries\fontsize{5.7}{6.1}\selectfont,text=gaNavy,align=center]
    at (1.15,4.10) {CONSTRAINED\\SENSORS};

  \foreach \x/\y in {0.85/3.35,1.65/2.85,0.85/2.25}{
    \fill[gaCyan,draw=gaBlue,line width=0.8pt] (\x,\y) circle (0.24);
    \fill[gaNavy] (\x,\y) circle (0.055);
    \draw[gaNavy,line width=0.75pt] (\x,\y+0.055)--(\x,\y+0.15);
    \draw[gaBlue,line width=0.7pt] (\x-0.12,\y+0.13)
      arc[start angle=210,end angle=-30,radius=0.14];
  }
  \draw[gaGray,dashed,line width=0.8pt] (1.10,3.32)--(2.65,3.02);
  \draw[gaGray,dashed,line width=0.8pt] (1.89,2.85)--(2.65,3.02);
  \draw[gaGray,dashed,line width=0.8pt] (1.08,2.30)--(2.65,3.02);
  \node[font=\fontsize{5.5}{6.0}\selectfont,text=gaGray,align=center]
    at (1.15,1.72) {Low-power devices\\and local radio links};
  \node[font=\bfseries\fontsize{5.3}{5.8}\selectfont,text=gaOrange,align=center]
    at (1.15,1.34) {Outside HPQ-AKE\\endpoint scope};

  \node[font=\bfseries\fontsize{5.7}{6.1}\selectfont,text=gaNavy,align=center]
    at (3.15,4.10) {IOT\\GATEWAY};
  \begin{scope}[shift={(3.15,2.92)}]
    \fill[gaNavy,rounded corners=0.8mm] (-0.48,-0.67) rectangle (0.48,0.67);
    \foreach \y in {-0.34,0,0.34}{\draw[white,line width=1.5pt] (-0.29,\y)--(0.29,\y);}
    \fill[gaGreen] (0.30,0.48) circle (0.055);
  \end{scope}
  \node[font=\bfseries\fontsize{6.2}{6.8}\selectfont,text=gaNavy,align=center]
    at (3.15,1.94) {RSA-enabled\\gateway-class node};
  \node[font=\fontsize{5.2}{5.8}\selectfont,text=gaGray,align=center]
    at (3.15,1.52) {Terminates AKE;\\aggregates sensor traffic};

  \node[font=\bfseries\fontsize{5.9}{6.4}\selectfont,text=gaNavy]
    at (6.10,4.15) {DUAL-KEM SECURE BACKHAUL};
  \draw[-{Stealth[length=2.2mm,width=1.6mm]},gaBlue,line width=1.15pt]
    (3.75,3.20)--(8.45,3.20);
  \draw[-{Stealth[length=2.2mm,width=1.6mm]},gaBlue,line width=1.15pt]
    (8.45,2.35)--(3.75,2.35);

  \node[rounded corners=1.2mm,draw=gaBlue,fill=gaCyan,line width=0.8pt,
    minimum width=3.35cm,minimum height=0.72cm] at (6.10,3.20) {};
  \node[font=\bfseries\fontsize{6.8}{7.3}\selectfont,text=gaNavy]
    at (6.10,3.34) {ML-KEM-768};
  \node[font=\fontsize{5.2}{5.7}\selectfont,text=gaGray]
    at (6.10,3.07) {PQ session secrecy + forward secrecy};

  \node[rounded corners=1.2mm,draw=gaGreen,fill=gaPaleGreen,line width=0.8pt,
    minimum width=3.35cm,minimum height=0.72cm] at (6.10,2.35) {};
  \node[font=\bfseries\fontsize{6.8}{7.3}\selectfont,text=gaNavy]
    at (6.10,2.49) {RSA-OAEP};
  \node[font=\fontsize{5.2}{5.7}\selectfont,text=gaGray]
    at (6.10,2.22) {Mutual implicit authentication};

  \node[rounded corners=1mm,draw=gaRed,fill=white,line width=0.75pt,
    font=\bfseries\fontsize{5.6}{6.1}\selectfont,text=gaRed]
    at (6.10,1.72) {NO PQ TRANSCRIPT SIGNATURES};

  \node[font=\bfseries\fontsize{5.9}{6.4}\selectfont,text=gaNavy]
    at (8.95,4.15) {EDGE / CLOUD};
  \begin{scope}[shift={(8.95,2.88)}]
    \fill[gaNavy,rounded corners=0.8mm] (-0.48,-0.67) rectangle (0.48,0.67);
    \foreach \y in {-0.34,0,0.34}{\draw[white,line width=1.5pt] (-0.29,\y)--(0.29,\y);}
  \end{scope}
  \node[font=\bfseries\fontsize{6.2}{6.8}\selectfont,text=gaNavy,align=center]
    at (8.95,1.94) {Authenticated\\edge/cloud service};
  \node[font=\fontsize{5.2}{5.8}\selectfont,text=gaGray,align=center]
    at (8.95,1.52) {Gateway-to-edge\\deployment target};

  \node[rounded corners=1.5mm,draw=gaPurple,fill=gaPalePurple,line width=0.8pt,
    minimum width=2.55cm,minimum height=3.6cm] at (11.48,2.82) {};
  \node[font=\bfseries\fontsize{5.0}{5.5}\selectfont,text=gaPurple]
    at (11.48,3.78) {FORMAL GUARANTEES};
  \foreach \yy/\txt in {3.48/AKE security,3.13/Forward secrecy,2.78/{\mbox{Conditional KCI}}}{
    \fill[gaGreen] (10.65,\yy) circle (0.08);
    \draw[white,line width=0.7pt] (10.61,\yy)--(10.64,\yy-0.04)--(10.70,\yy+0.05);
    \node[anchor=west,font=\fontsize{5.7}{6.2}\selectfont,text=gaNavy] at (10.82,\yy) {\txt};
  }
  \draw[gaPurple!45,line width=0.55pt] (10.55,2.48)--(12.41,2.48);
  \node[font=\fontsize{5.2}{5.8}\selectfont,text=gaGray,align=center]
    at (11.48,2.15) {ROM authentication\\QROM session secrecy};
  \node[font=\bfseries\fontsize{5.3}{5.8}\selectfont,text=gaPurple,align=center]
    at (11.48,1.72) {M-LWE + RSA\\assumptions};

  \node[rounded corners=1.2mm,draw=gaOrange,fill=gaPaleOrange,line width=0.8pt,
    minimum width=3.92cm,minimum height=1.08cm] at (2.25,0.62) {};
  \node[font=\bfseries\fontsize{12.5}{12.5}\selectfont,text=gaOrange] at (2.25,0.80) {56.4\%};
  \node[font=\bfseries\fontsize{5.0}{5.5}\selectfont,text=gaNavy] at (2.25,0.49)
    {LESS HANDSHAKE TRANSMISSION};
  \node[font=\fontsize{5.2}{5.7}\selectfont,text=gaGray] at (2.25,0.25)
    {13,009 $\rightarrow$ 5,668 bytes};

  \node[rounded corners=1.2mm,draw=gaPurple,fill=gaPalePurple,line width=0.8pt,
    minimum width=3.92cm,minimum height=1.08cm] at (6.50,0.62) {};
  \node[font=\bfseries\fontsize{12.5}{12.5}\selectfont,text=gaPurple] at (6.50,0.80) {31.4\%};
  \node[font=\bfseries\fontsize{5.0}{5.5}\selectfont,text=gaNavy] at (6.50,0.49)
    {LOWER MODELED LATENCY};
  \node[font=\fontsize{5.2}{5.7}\selectfont,text=gaGray] at (6.50,0.25)
    {2,791 $\rightarrow$ 1,914 ms at 50 kbps};

  \node[rounded corners=1.2mm,draw=gaGreen,fill=gaPaleGreen,line width=0.8pt,
    minimum width=3.92cm,minimum height=1.08cm] at (10.75,0.62) {};
  \node[font=\bfseries\fontsize{12.5}{12.5}\selectfont,text=gaGreen] at (10.75,0.80) {7.11 ms};
  \node[font=\bfseries\fontsize{5.0}{5.5}\selectfont,text=gaNavy] at (10.75,0.49)
    {MEASURED LOCAL COMPUTATION};
  \node[font=\fontsize{5.2}{5.7}\selectfont,text=gaGray] at (10.75,0.25)
    {x86 gateway-class testbed};
\end{tikzpicture}}\par
\medskip

\end{graphicalabstract}

\begin{highlights}
\item HPQ-AKE replaces post-quantum transcript signatures with dual KEMs.
\item ML-KEM-768 provides secrecy; RSA-OAEP authenticates IoT gateways.
\item \rev{Analysis covers AKE, PFS, and KCI under explicit exposure assumptions.}
\item Handshake transmission falls by 56.4\%, from 13,009 to 5,668 bytes.
\item Modeled latency falls by 31.4\% on a 50 kbps satellite-IoT link.
\end{highlights}

\begin{keywords}
Post-Quantum Authentication \sep Internet of Things Security \sep Hybrid Authenticated Key Exchange \sep KEM-Based Implicit Authentication \sep Provable Security \sep Forward Secrecy
\end{keywords}

\maketitle

\section{Introduction}

{T}he security infrastructure of the modern Internet relies heavily on traditional public key cryptosystems,
such as RSA and Elliptic Curve Cryptography (ECC). However, the rapid advancement of large-scale
quantum computing poses a long-term risk to these foundations. Shor’s algorithm, published in 1994,
demonstrated the ability to solve integer factorization and discrete logarithm problems in polynomial
time, effectively rendering current standard protocols insecure \cite{Shor}. In response to this threat, which motivates the migration to post-quantum cryptography, the National Institute of Standards and Technology (NIST) initiated the Post-Quantum Cryptography (PQC) standardization process. By 2022, NIST announced the selected algorithms, with CRYSTALS-Kyber emerging as the primary standard for Key Encapsulation Mechanisms (KEM) due to its superior
performance and strong lattice-based security guarantees \cite{nist-report2022,Avanzi}.
While the problem of confidentiality has been addressed by Kyber, post-quantum authentication
remains a significant challenge. The digital signature schemes selected by NIST, such as CRYSTALS-Dilithium and Falcon, often have larger public keys or signatures than RSA or ECC, or incur higher computational costs \cite{Ducas}. This creates a bottleneck for resource-constrained devices
and latency-sensitive networks. Furthermore, the gradual replacement of the Public Key Infrastructure (PKI) that currently manages billions of RSA certificates is a costly, decades-long endeavor.
In August 2024, the NIST officially finalized the
first set of post-quantum standards, where CRYSTALS-Kyber was standardized as ML-KEM in FIPS
203 \cite{nist-fips203-Alagic}. While this marks a significant milestone, the industry is increasingly leaning towards hybrid
constructions to mitigate the risks of potential vulnerabilities in early PQC implementations. This trend
is evidenced by the ongoing standardization efforts within the Internet Engineering Task Force (IETF) to integrate hybrid key exchange
mechanisms into the TLS 1.3 protocol \cite{Stebila}. However, the communication overhead introduced by post-quantum certificate chains and large signature payloads remains a major deployment concern for high-latency or severely degraded links such as satellite communications, wireless IoT infrastructures, and intermittently connected edge networks.
During this transition period, hybrid authenticated key exchange (AKE) protocols are considered
the most pragmatic approach. These protocols aim to combine post-quantum confidentiality (to protect
against “harvest now, decrypt later” attacks) with classical authentication mechanisms \cite{7-Bos}. However, most existing hybrid proposals (e.g., Hybrid TLS 1.3) still follow the traditional Signed-Diffie-Hellman paradigm, where large post-quantum signatures and certificate chains substantially increase handshake communication overhead and transmission latency, particularly in bandwidth-constrained environments.
To address this efficiency gap, the concept of “KEM-based authentication” (or implicit authentication) has been proposed as an effective alternative to digital signatures, most notably demonstrated by
the KEMTLS protocol \cite{kemtls}. The core idea is to use long-term keys to encapsulate a shared secret rather
than signing the transcript, thereby reducing communication overhead and improving bandwidth efficiency during authenticated handshakes. However,
prior research has largely focused on pure post-quantum replacements, often overlooking the potential
to leverage existing RSA hardware acceleration to optimize performance during the hybrid phase.
In this paper, we propose \textbf{HPQ-AKE} (Hybrid Post-Quantum Authenticated Key Exchange), a novel ``sign-less" hybrid protocol designed to bypass the limitations of post-quantum signatures. \rev{By using ML-KEM-768 for session secrecy and RSA-OAEP (Optimal Asymmetric Encryption Padding) as a KEM for implicit authentication, HPQ-AKE targets bandwidth-efficient migration for RSA-enabled gateways and edge/cloud services, subject to compatible key-use policies and protocol support at both endpoints.}

\noindent\textbf{\rev{Industrial motivation and deployment constraints.}}
\rev{Commercial edge platforms illustrate the gateway-assisted architecture considered here: AWS IoT Greengrass supports local processing on devices, hubs, or gateways with intermittent cloud connectivity, while its authentication model uses X.509 certificates for local and cloud connections \mbox{\cite{AWSGreengrassFAQ,AWSGreengrassAuthentication}}. A concrete industrial example is Ergon Energy's documented use of Inmarsat BGAN M2M to monitor and control remote electricity-network reclosers where terrestrial connectivity is limited or unavailable; the case study reports 5--15 MB of traffic per recloser per month \mbox{\cite{InmarsatErgonEnergy}}. These examples motivate investigating smaller authenticated handshakes for remote gateway backhaul with constrained communication budgets; they are not deployments or validations of HPQ-AKE. In particular, the 50 kbps, 600 ms RTT scenario evaluated here is a synthetic stress case, not a measurement of either system.}

\rev{Our deployment boundary is the gateway-to-edge/cloud link, not the sensor-to-gateway radio segment. The gateway is a trusted security endpoint, so this architecture does not itself provide sensor-to-cloud end-to-end protection through an untrusted gateway. Each HPQ-AKE endpoint must support one RSA-3072 private decapsulation per completed handshake, protected long-term key storage, and authenticated peer-key provisioning and renewal. Reconnection can make handshake overhead relevant, but whether smaller messages reduce latency depends on backhaul bandwidth, RTT, certificate caching, and available computation. The present evidence comes from x86 micro-benchmarks and network modeling, not industrial field trials or direct microcontroller execution.}

To formally guide our architectural design, security analysis, and empirical evaluation, this paper addresses the following three fundamental research questions:
\begin{itemize}
    \item \textbf{RQ1 (\textit{Rigorous Provable Security})}: How can a dual-layer KEM-based hybrid protocol (combining Kyber and RSA-OAEP) be formally proven to guarantee authenticated key establishment and perfect forward secrecy (PFS) within a reductionist complexity-theoretic framework?
    \item \textbf{RQ2 (\textit{Bandwidth-Optimized Architecture})}: \rev{How can a sign-less hybrid AKE reuse suitable RSA-based PKI infrastructure while minimizing handshake communication on bandwidth-constrained gateway-to-edge/cloud links?}
    \item \textbf{RQ3 (\textit{Empirical Performance})}: To what extent does the proposed HPQ-AKE protocol reduce handshake latency and packet transmission volume compared to state-of-the-art hybrid standards (e.g., Hybrid TLS 1.3) when evaluated through empirical cryptographic micro-benchmarks and simulated bandwidth-constrained channels such as satellite-linked IoT network models?\par
\end{itemize}

\noindent \rev{The following contributions map each research question to its core conclusion and the corresponding analysis:}
\begin{itemize}
    \item \textbf{\rev{RQ1---Security conclusion}}: \rev{The extended Bellare--Rogaway analysis establishes authenticated key establishment, PFS, and conditional KCI resistance under the long-term-key-only exposure assumptions. Section~\mbox{\ref{sec:security_analysis}} separates classical RSA-based authentication in the ROM from ML-KEM-based session secrecy and forward secrecy in the QROM. Thus, the result supports transitional post-quantum confidentiality, not authentication against an active adversary capable of breaking RSA.}

    \item \textbf{\rev{RQ2---Architecture conclusion}}: \rev{The dual-KEM design in Section~\mbox{\ref{sec:protocol}} replaces transcript signatures with ML-KEM-768 session-key establishment and RSA-OAEP implicit authentication. Under the payload accounting in Table~\mbox{\ref{tab:payload_decomposition}}, total handshake transmission decreases from 13,009 to 5,668 bytes, a 56.4\% reduction relative to the modeled Hybrid TLS 1.3 Full baseline. Reuse concerns suitable RSA credentials and infrastructure; it does not imply unchanged TLS interoperability or removal of certificate provisioning and validation requirements.}

    \item \textbf{\rev{RQ3---Performance conclusion}}: \rev{Section~\mbox{\ref{sec:performance}} measures 7.11 ms of total local computation on the x86 testbed. In the simulated 50 kbps, 600 ms RTT scenario, median handshake latency decreases from 2,791 to 1,914 ms relative to the Full baseline, a 31.4\% reduction, but remains above the Pre-cached baseline. At the separate 20 ms RTT crossover setting, the modeled break-even bandwidths are 0.31 Mbps against the Pre-cached baseline and 4.08 Mbps against the Full baseline. These conditional results identify a bandwidth--computation trade-off rather than a universal latency advantage or demonstrated MCU feasibility.}
\end{itemize}

\noindent \textbf{Paper Organization.} The remainder of this paper is organized as follows: Section \ref{sec:relatedwork} reviews related work. Section \ref{sec:preliminaries} introduces the necessary cryptographic primitives and mathematical foundations. Section \ref{sec:protocol} describes the proposed HPQ-AKE protocol in detail. Section \ref{sec:security_analysis} provides the security analysis. Section \ref{sec:performance} presents the performance evaluation, and Section \ref{sec:conclusion} concludes the paper.

\section{Related Work}
\label{sec:relatedwork}
Research on securing communication channels against quantum adversaries has advanced significantly, focusing on standardization, hybrid real-world deployments, and efficient authentication mechanisms.

\noindent\textbf{\rev{Industrial gateway deployment and migration requirements.}}
\rev{The Greengrass documentation and Ergon Energy case study provide complementary practical context: the former describes certificate-authenticated edge/cloud operation, whereas the latter documents remote utility connectivity over satellite \mbox{\cite{AWSGreengrassFAQ,AWSGreengrassAuthentication,InmarsatErgonEnergy}}. Neither establishes support for HPQ-AKE or the suitability of installed RSA keys for OAEP. Certificate key-usage restrictions must be respected: RFC 5280 distinguishes signature use from key encipherment \mbox{\cite{rfc5280KeyUsage}}. Accordingly, a candidate deployment must verify RSA-OAEP support and permitted key usage, provision appropriate credentials when existing keys are signing-only, and implement HPQ-AKE at both backhaul endpoints. Peer-key distribution, renewal, revocation handling during outages, and gateway trust remain operational requirements, not problems solved by reducing handshake bytes. This motivates HPQ-AKE as a candidate for managed RSA-enabled backhaul migration rather than a drop-in replacement for commercial TLS services.}

\noindent \textbf{Formal Security Models for Authenticated Key Exchange.}
The formal analysis of AKE protocols follows several established models. The Bellare--Rogaway (BR) framework \cite{b1} defines security through the Session Key Indistinguishability (SK-Security) game, modeling adversaries with \textsf{Send}, \textsf{Reveal}, and \textsf{Corrupt} oracles. The extended Canetti--Krawczyk (eCK) model \cite{LaMacchia2007, Canetti2001} captures stronger adversarial capabilities, including combined ephemeral and static key exposure, providing a natural framework for analyzing protocols with forward secrecy requirements. For post-quantum settings, the Quantum Random Oracle Model (QROM) extends classical ROM proofs to account for adversaries capable of making quantum superposition queries, as required for protocols combining lattice-based and classical primitives \cite{BonehZhandry, zhandry2015secure}.
Our analysis of HPQ-AKE employs a Hybrid Oracle Model Framework that splits authentication and secrecy analysis: classical authentication properties are proven in the ROM under RSA hardness, while forward secrecy against quantum adversaries is bounded in the QROM under M-LWE hardness. This hybrid proof strategy directly reflects the transitional nature of the protocol's threat model and is more rigorous than proofs that apply either model uniformly to hybrid constructions. 

\noindent \textbf{Post-Quantum Standardization and IoT Challenges.}
Following the NIST PQC standardization process initiated in 2016, CRYSTALS-Kyber was selected as the standard for KEM \cite{nist-report2022, Avanzi}. For digital signatures, CRYSTALS-Dilithium and Falcon were chosen \cite{Ducas}. While these algorithms are secure, their implementation on resource-constrained devices poses significant challenges. Recent benchmarks on embedded systems, such as the ARM Cortex-M4, indicate that post-quantum signatures often incur higher energy consumption and memory usage than their encapsulation counterparts \cite{pqm4-Kannwischer}. This asymmetry creates a bottleneck for IoT environments, motivating the search for lighter authentication alternatives. 

\noindent \textbf{IoT-Oriented AKE and Post-Quantum Evaluation.}
EDHOC provides a compact authenticated key exchange for constrained IoT and OSCORE deployments \cite{rfc9528}, and recent work has evaluated EDHOC-based operation over low-power networks \cite{AriasJimenez2025}. PUF-based MAKEP protocols address lightweight thing-to-server and thing-to-thing authentication \cite{Zerrouki2023}, while Cruz-Piris et al. quantify the system-level impact of PQC in Industrial IoT \cite{CruzPiris2025}. Hanna et al. provide a comprehensive performance evaluation of post-quantum security for consumer IoT devices \cite{Hanna2025}. In contrast, HPQ-AKE contributes a sign-less AKE protocol with mutual implicit authentication, a formal security reduction, and communication-centric break-even analysis rather than only evaluating existing post-quantum mechanisms. 

\noindent \textbf{Hybrid Key Exchange in Practice.}
To mitigate the risks associated with early PQC deployments, the industry has increasingly adopted hybrid constructions combining classical and post-quantum primitives. Stebila and Mosca \cite{7-Bos} established the theoretical foundations of hybrid TLS designs, while large-scale industrial experiments, such as Google's and Cloudflare's CECPQ2, demonstrated the practical feasibility of combining X25519 with lattice-based KEMs. More recently, the IETF has standardized hybrid key exchange mechanisms integrating ML-KEM into TLS~1.3 workflows \cite{Stebila}. However, most existing hybrid approaches preserve the traditional Signed-Diffie-Hellman authentication paradigm, where post-quantum signatures remain responsible for authentication. Although this strategy simplifies migration compatibility, it introduces substantial communication overhead due to large certificate chains and signature payloads. Recent empirical evaluations have shown that in bandwidth-constrained environments, transmission latency increasingly dominates local cryptographic processing costs, particularly under high-RTT satellite or IoT conditions. Consequently, reducing handshake payload size has become an equally important optimization objective alongside computational efficiency. 

\noindent \textbf{KEM-based Authentication (Sign-less Protocols).}
A promising approach to eliminating signature overhead is ``implicit authentication" or KEM-based authentication. This concept utilizes long-term KEM keys to mutually authenticate peers via successful key decapsulation. The KEMTLS protocol by Schwabe et al. \cite{kemtls} successfully integrated this concept into TLS 1.3, showing that authentication can be achieved with reduced bandwidth. Following this, other variants, such as KEMTLS-PDK (Pre-Distributed Keys), have been proposed to further optimize latency. \rev{These protocols differ in both trust setup and cryptographic composition. For example, Muckle\# combines classical, post-quantum, and quantum-key-distribution inputs with post-quantum KEM-based authentication \mbox{\cite{9-Battarbee}}, whereas other hybrid approaches use identity-based encryption \mbox{\cite{Astrizi}}.} Our proposed HPQ-AKE leverages the ubiquity of RSA-OAEP hardware acceleration while targeting bandwidth-constrained deployment environments. Unlike prior works that primarily minimize computational complexity, HPQ-AKE adopts a communication-centric design philosophy. It trades a total computational footprint of \textbf{7.11 ms} for a \textbf{56.4\% reduction} in total handshake payload size (shrinking transmitted data from \textbf{13,009 Bytes to 5,668 Bytes} compared to the Hybrid TLS 1.3 Full handshake). As established by our break-even analysis, this design becomes advantageous in satellite, wireless, and edge IoT \rev{networks below the 0.31 Mbps threshold at the modeled 20 ms RTT setting; larger RTT values lower this threshold.}
\par

\noindent\textbf{\rev{Trust-model distinction.}}
\rev{HPQ-AKE assumes trusted identity-to-key bindings before the handshake: the initiator stores the responder's RSA and static ML-KEM public keys, and the responder stores the initiator's RSA public key. Enrollment may use an enterprise PKI or another authenticated provisioning channel; cache integrity, key updates, and revocation handling remain external requirements. In contrast, the full KEMTLS handshake authenticates the peer's KEM key through a CA-signed certificate chain validated against a local trust anchor \mbox{\cite{kemtls}}. Removing transcript signatures does not remove certificate signatures, and certificate validation does not inherently require contacting a CA during each handshake. KEMTLS-PDK also supports pre-distributed or cached keys \mbox{\cite{Schwabe2021}} and is therefore closer to HPQ-AKE in its prior-trust requirement. HPQ-AKE shifts identity-to-key validation to provisioning rather than eliminating it: its reported handshake savings exclude enrollment and credential maintenance. Consequently, comparison with full KEMTLS involves different credential-delivery assumptions, and our Pre-cached Hybrid TLS baseline is not an evaluation of KEMTLS-PDK.}

\noindent \textbf{Revisiting RSA for Hybrid Authentication.}
Our work diverges from the pure PQC trend by revisiting RSA-KEM, formally proposed by Shoup \cite{V.Shoup}. Unlike the vulnerable RSA-PKCS v1.5, RSA-KEM (instantiated via RSA-OAEP) is secure in the Random Oracle Model and allows for encryption-based authentication. This design choice addresses the critical security communication trade-off inherent in the post-quantum transition. While purely lattice-based signatures provide quantum resistance, their large certificate chains and transmitted payload sizes can become prohibitive in latency-sensitive or bandwidth-limited environments, even when local computation remains relatively efficient. \rev{HPQ-AKE combines RSA-based implicit authentication with ML-KEM session secrecy to investigate a communication-oriented migration trade-off. Its timing evidence is limited to the reported x86 software testbed and the associated network model. RSA hardware acceleration remains a potential deployment option, not a benefit established by these measurements; performance on ARM gateways and unaccelerated microcontrollers requires platform-specific evaluation.}

To clearly distinguish our contributions from existing literature, Table \ref{tab:comparison} provides a multi-dimensional comparison across key security and performance criteria.

\begin{table*}[pos=t]
\centering
\caption{Comparison of Related Works and the Proposed Protocol}
\label{tab:comparison}
\small
\begin{tabularx}{\textwidth}{l >{\RaggedRight\arraybackslash}p{5cm} >{\RaggedRight\arraybackslash}X}
\toprule
\textbf{Reference} & \textbf{Main Contribution} & \textbf{Unaddressed Aspects} \\ \midrule

KEMTLS (2020) \cite{kemtls} & \rev{Uses a certified long-term KEM key instead of a transcript signature for server authentication.} & \rev{The full handshake retains CA-signed certificates and TLS-PKI validation; unlike HPQ-AKE, it does not assume the peer's leaf key is already trusted locally.} \\ \addlinespace

EDHOC (RFC 9528) \cite{rfc9528,AriasJimenez2025} & Provides a compact AKE optimized for sensor-level constrained IoT and low-power deployments. & Targets device-level exchanges rather than gateway-level post-quantum migration and does not define the ML-KEM + RSA-OAEP transitional architecture of HPQ-AKE. \\ \addlinespace

Alagic et al. (2025) \cite{Alagic2025} & Formalizes secure combiners to ensure hybrid KEMs remain IND-CCA secure. & Purely theoretical/mathematical study. Does not propose an interactive AKE protocol with specific entity exchange steps. \\ \addlinespace

Ricci et al. (2024) \cite{Ricci} & Implements a 3-key hybrid system (Classical + PQC + QKD) on FPGA hardware. & High architectural complexity due to QKD requirements. Not practical for standard internet environments without specialized hardware. \\ \addlinespace

Amal Raj (2025) \cite{Raj} & Introduces a hybrid framework combining EdDSA, ML-DSA, and ML-KEM. & Still relies on ML-DSA signatures for authentication, resulting in substantial certificate-chain overhead and larger handshake payloads. \\ \addlinespace

Crockett et al. (2019) \cite{E.Crockett} & Reports on early implementations of PQC/Hybrid key exchanges in TLS and SSH. & Focuses on using KEMs for key exchange but still uses signatures for authentication. Lacks dedicated KEM-based authentication logic. \\ \addlinespace

Battarbee et al. (2025) \cite{9-Battarbee} & \rev{Proposes Muckle\#, replacing Muckle+'s signatures with post-quantum KEM authentication.} & \rev{Combines classical, post-quantum, and QKD key material; its scope differs from HPQ-AKE's RSA-enabled gateway migration.} \\ \addlinespace

Zhang et al. (2015) \cite{Zhang} & Provides a signature-free 2-pass AKE based on the hardness of Ring-LWE. & Uses implicit authentication and is purely post-quantum, lacking hybrid compatibility for the current transition period. \\ \addlinespace

Xia et al. (2024) \cite{Xia} & Proposes a Kyber-based scheme specifically for UAV networks. & Requires a Trusted Authority (TA) for registration. Context-specific and not a general-purpose hybrid RSA+Kyber protocol. \\ \addlinespace

Schwabe et al. (2021) \cite{Schwabe2021} & \rev{KEMTLS-PDK (pre-distributed keys) uses previously obtained public keys to reduce handshake communication and computation.} & \rev{Shares HPQ-AKE's need for trusted prior key distribution. Its TLS-derived handshake differs from HPQ-AKE's ML-KEM + RSA-OAEP design; pre-caching is not unique to HPQ-AKE.} \\ \addlinespace

RFC 5990 (RSA-KEM) \cite{rfc5990} & Specifies the use of RSA-KEM as a key transport algorithm in CMS. & A non-interactive 1-pass mechanism. Lacks mutual authentication, Forward Secrecy (PFS), and quantum resistance. \\ \addlinespace

Zhou \& Feng \cite{Zhou} & Surveys side-channel attack techniques and physical security testing. & A survey/position paper. Identifies implementation risks but does not provide a logic for authenticated key exchange. \\ \bottomrule

\end{tabularx}
\par\smallskip
\begin{minipage}{\textwidth}
\footnotesize
\rev{Terminology: CA = certificate authority; IND-CCA = indistinguishability under chosen-ciphertext attack; QKD = quantum key distribution; FPGA = field-programmable gate array; UAV = unmanned aerial vehicle; CMS = Cryptographic Message Syntax, a container format for protected messages.}
\end{minipage}
\end{table*}
Unlike prior evaluations that primarily report isolated cryptographic execution times, our work introduces a multi-tiered validation framework combining empirical hardware micro-benchmarking, physics-compliant multi-RTT analytical latency modeling, and stochastic Monte Carlo network simulations executed over \textbf{10,000 algorithmic iterations}. This enables a statistically rigorous assessment of protocol stability, exact break-even bandwidth boundaries, and end-to-end latency variance under realistic, jitter-prone deployment environments.

\section{Cryptographic Primitives and Mathematical Foundations}
\label{sec:preliminaries}
In this section, we establish the mathematical notations and cryptographic primitives utilized in our protocol. We denote the security parameter by $\lambda$. For a finite set $S$, $x \xleftarrow{\$} S$ denotes sampling $x$ uniformly at random from $S$. We use $||$ for string concatenation and $\oplus$ for bitwise XOR. A function $\nu: \mathbb{N} \to \mathbb{R}$ is negligible, denoted $\nu(\lambda) \le \mathsf{negl}(\lambda)$, if it vanishes faster than the inverse of any polynomial. These primitives are selected to balance post-quantum confidentiality, communication efficiency, deployment compatibility, and computational practicality in hybrid migration environments. In particular, lattice-based cryptography provides quantum-resistant session secrecy, while RSA-based authentication leverages widely deployed hardware accelerators and compact certificate infrastructures to reduce handshake communication overhead during the transition period \cite{V.Shoup}.

\subsection{Lattice-Based Cryptography}
Our protocol's post-quantum confidentiality relies on the hardness of lattice problems over module rings. Lattice-based cryptography is favored in post-quantum standards due to its resistance to quantum attacks, such as those using Grover's or Shor's algorithms, and its efficient performance on classical hardware \cite{Shor}.

\subsubsection{Module-LWE Problem}
Let $R_q = \mathbb{Z}_q[X]/(X^n+1)$ be the ring of polynomials where $n$ is a power of 2 and $q$ is a prime. The Module Learning With Errors (M-LWE) problem serves as the hardness assumption for CRYSTALS-Kyber \cite{Avanzi}. This problem is believed to be quantum-hard, with no known efficient quantum algorithms to solve it, making it suitable for post-quantum security \cite{BosEtAl}.

\begin{definition}[M-LWE Distribution]
Let $n, q, k$ be integers and $R_q = \mathbb{Z}_q[X]/(X^n+1)$. For a secret vector $\mathbf{s} \in R_q^k$ and error vectors $\mathbf{e} \in R_q^k$ sampled from a centered binomial distribution $\beta_\eta$, and a matrix $\mathbf{A} \in R_q^{k \times k}$ sampled uniformly, an M-LWE sample is $(\mathbf{A}, \mathbf{t} := \mathbf{A}\mathbf{s} + \mathbf{e})$. The {\rm{Decision M-LWE assumption}} states that $(\mathbf{A}, \mathbf{t})$ is computationally indistinguishable from $(\mathbf{A}, \mathbf{u})$ where $\mathbf{u} \xleftarrow{\$} R_q^k$.
\end{definition}

\begin{definition}[Decision M-LWE Hardness]
The {\rm{Decision M-LWE problem}} asks to distinguish pairs $(\mathbf{A}, \mathbf{t})$ drawn from the M-LWE distribution from pairs $(\mathbf{A}, \mathbf{u})$ where $\mathbf{u} \xleftarrow{\$} R_q^k$ is uniform. We assume this problem is computationally intractable for any quantum polynomial-time adversary.
\end{definition}

\subsection{Key Encapsulation Mechanisms (KEM)}
A KEM scheme $\mathcal{K} = (\mathsf{KeyGen}, \mathsf{Encaps}, \mathsf{Decaps})$ allows two parties to derive a shared secret. We require IND-CCA security for all KEMs used in HPQ-AKE. IND-CCA provides security against adaptive chosen-ciphertext attacks, which are critical for preventing decryption oracles in key exchange protocols \cite{b1}.
\subsubsection{Post-Quantum KEM: Kyber}
We utilize Kyber, an IND-CCA secure KEM constructed from an IND-CPA PKE scheme via the Fujisaki-Okamoto transform. Kyber has been standardized as ML-KEM in FIPS 203 by NIST in 2024, providing quantum-resistant key encapsulation with parameters optimized for security levels 1, 3, and 5 \cite{Avanzi}. We use ML-KEM-768 (formerly Kyber-768) to balance post-quantum security, compact ciphertext size, and low computational overhead, offering approximately 192 bits of post-quantum security \cite{nistfourthround}:
\begin{itemize}
    \item $\mathsf{KeyGen}_{PQ}()$: Generates $(pk_{pq}, sk_{pq})$.
    \item $\mathsf{Encaps}_{PQ}(pk_{pq}) \to (K, C)$: Returns a shared key $K$ and ciphertext $C$.
    \item $\mathsf{Decaps}_{PQ}(sk_{pq}, C) \to K$: Recovers $K$ or returns implicit rejection $\bot$.
\end{itemize}

\subsubsection{RSA-KEM with OAEP}
To leverage existing RSA-enabled infrastructures and compact certificate ecosystems, we instantiate a KEM \cite{V.Shoup} using RSA-OAEP \cite{ShoupOAEP}. While standard RSA is a public key encryption scheme, its construction under the OAEP padding provides the required IND-CCA security in the Random Oracle Model:
\begin{itemize}
\item $\mathsf{Encaps}_{RSA}(pk_{rsa})$: Samples $r \xleftarrow{\$} \{0,1\}^{\lambda}$, computes $C \leftarrow \mathsf{Enc}(pk_{rsa}, r)$, and derives $K \leftarrow H(r)$.
\item $\mathsf{Decaps}_{RSA}(sk_{rsa}, C)$: Recovers $r \leftarrow \mathsf{Dec}(sk_{rsa}, C)$. If $r \neq \bot$, returns $K \leftarrow H(r)$, otherwise returns $\bot$.
\end{itemize}
The security of RSA-KEM relies on the \textbf{RSA Assumption}: Given $(N, e, y)$, finding $x$ such that $x^e \equiv y \pmod N$ is hard. RSA-OAEP is provably secure in the Random Oracle Model (ROM) against chosen-ciphertext attacks, as specified in PKCS\#1 v2.2 (RFC 8017) \cite{rfc8017}. While RSA is vulnerable to quantum attacks via Shor's algorithm, in our hybrid design, it is used solely for authentication during the transition period, assuming no large-scale quantum computers are available yet.

\subsubsection{Hybrid Rationale} Our construction follows a communication-aware hybrid design philosophy in which post-quantum confidentiality is anchored in the M-LWE assumption, while authentication compatibility is preserved through RSA-based mutual implicit authentication. Unlike conventional hybrid TLS constructions that inherit the large communication footprint of post-quantum signatures, HPQ-AKE combines Kyber with RSA-KEM to eliminate certificate-chain and signature overhead, compressing the total transmitted handshake payload size by \textbf{56.4\%} (from \textbf{13,009 Bytes down to 5,668 Bytes}). Consequently, the architecture trades a bounded total computational footprint of \textbf{7.11 ms} for lower transmission latency in severely bandwidth-constrained environments \rev{operating below the 0.31 Mbps break-even boundary at 20 ms RTT; this boundary depends on RTT.}

\subsection{Authenticated Encryption with Associated Data (AEAD)}
To replace digital signatures, our protocol utilizes an AEAD scheme $\Pi_{AEAD} = (\mathsf{Enc}, \mathsf{Dec})$ to bind the session transcript. AEAD provides both confidentiality and integrity while avoiding the communication overhead associated with large post-quantum signatures, making it suitable for lightweight sign-less authentication, as in our sign-less approach \cite{rfc5116}. We instantiate AEAD with AES-GCM for its efficiency and hardware support \cite{McGrewViegaGCM}.
\par
\begin{definition}[AEAD Security]
An AEAD scheme is {\rm{secure}} if it provides:
\begin{enumerate}
    \item \textbf{Privacy (IND-CPA)}: Ciphertexts are indistinguishable from random strings.
    \item \textbf{Integrity (INT-CTXT)}: It is computationally infeasible to produce a valid ciphertext not output by the encryption oracle.
\end{enumerate}
\end{definition}
In our protocol, $\mathsf{Enc}(K, N, AD, M) \to C$ takes a key $K$, nonce $N$, associated data $AD$, and message $M$. The decryption $\mathsf{Dec}(K, N, AD, C)$ returns $M$ or error symbol $\bot$. \rev{We specify AES-256-GCM with a 96-bit nonce and a 128-bit appended authentication tag. RFC 5116 requires nonce uniqueness for encryption under each particular key; it does not make AES-GCM resistant to nonce misuse. Reusing a key--nonce pair for different plaintexts can compromise confidentiality and integrity \mbox{\cite{rfc5116,nist800-38D}}.}

\noindent\textbf{\rev{Handshake nonce policy.}}
\rev{The two AEAD encryptions use implicit, direction-specific 12-byte constants:}
\[
\mbox{$\nu_T = 0^{96}, \qquad \nu_C = 0^{95} \parallel 1.$}
\]
\rev{The token uses $\nu_T$ with the fresh key $k_{RSA}^A$; the challenge uses $\nu_C$ with the independently generated fresh key $k_{RSA}^B$. Receivers use the same constants for decryption. These are nonempty nonce values, distinct from the 256-bit session nonces $N_A$ and $N_B$, and require no transmitted nonce field. Each encapsulated key is permitted at most one AEAD encryption and must not be reused for another session or application-data encryption. The long-term RSA key may be reused for encapsulation/decapsulation; the one-encryption restriction applies to the encapsulated symmetric key.}

\rev{Thus, conditional on fresh independent key generation and absence of accidental key collisions, no two honest encryption invocations repeat a key--nonce pair, even though the nonce constants recur across different keys. Retransmission must reuse the cached ciphertext byte-for-byte, not call Seal again. A changed payload, restarted handshake, or recovery after a crash requires new encapsulation and fresh session state; restoring used keys or rolled-back randomness is prohibited. Repeated decryption does not violate the encryption-nonce requirement, and protocol replay checks remain necessary. This is a nonce-respecting specification, not a claim of security if its key-use discipline is violated.}
\subsection{Cryptographic Hash Functions}
We employ cryptographic hash functions $H: \{0,1\}^* \to \{0,1\}^d$ (SHA-3) and a Key Derivation Function (KDF) modeled as HKDF \cite{rfc5869}.
In the security analysis, we model these functions as \textbf{Random Oracles (ROM)}. We assume they satisfy:
\begin{itemize}
    \item \textbf{Collision Resistance}: Hard to find $x \neq x'$ such that $H(x) = H(x')$.
    \item \textbf{Pre-image Resistance}: Given $y$, hard to find $x$ such that $H(x) = y$.
\end{itemize}
These properties are crucial for the integrity of the Session ID ($SID$) and the binding of authentication tokens. The compact transcript-binding mechanism additionally avoids the transmission overhead associated with explicit signature-based authentication. SHA-3 is specified in FIPS 202, offering quantum resistance against Grover's algorithm (e.g., 128-bit security requires $2^{64}$ quantum queries) \cite{nistSHA-3}. HKDF (RFC 5869) is used for key derivation to ensure pseudorandom outputs from shared secrets \cite{rfc5869}.

\section{Proposed Protocol}
\label{sec:protocol}
We propose HPQ-AKE, a Hybrid Post-Quantum Authenticated Key Exchange protocol designed to operate efficiently in constrained communication environments such as low-bandwidth wireless and satellite links. HPQ-AKE departs from the traditional Signed-Diffie-Hellman paradigm by eliminating digital signatures entirely. Instead, it employs a \textit{Dual-Layered KEM} architecture:
\begin{enumerate}
    \item \textbf{Quantum Layer}: Uses ML-KEM (Kyber) \cite{Avanzi} for ephemeral key exchange, ensuring post-quantum confidentiality and forward secrecy for the established session key under the M-LWE assumption.
    \item \textbf{Classical Layer}: Uses RSA-OAEP \cite{Fujisaki} as a Key Encapsulation Mechanism (RSA-KEM) to provide implicit authentication, while explicit key confirmation is achieved in the final confirmation phase.
\end{enumerate}

The protocol follows a three-pass design shown in Fig. \ref{fig:protocol_flow_large}. The static Kyber encapsulation provides an authenticated post-quantum anchor toward the responder, while the ephemeral encapsulation contributes fresh entropy to achieve forward secrecy. Authenticated Encryption with Associated Data (AEAD), specifically AES-GCM, is used to secure the handshake payloads. Since the encrypted payloads are small (authentication tokens and challenge values only), the symmetric encryption overhead is negligible compared to the asymmetric KEM operations in our experimental evaluation.

\subsection{IoT Deployment Model}
HPQ-AKE targets a gateway-assisted IoT architecture. Constrained sensors or actuators communicate through a local access link and are not assumed to execute the complete handshake in the present evaluation. An RSA-enabled gateway acts as Initiator A and terminates security on the device side, while an edge/cloud service acts as Responder B. Each cryptographic endpoint performs one RSA-3072 private decapsulation per completed handshake. \rev{Gateway-class and server-class nodes with adequate RSA capability are the intended deployment targets, not hardware classes validated by the present measurements. Here, RSA-enabled denotes support for RSA-OAEP operations, not verified dedicated RSA acceleration.} Direct execution on an unaccelerated microcontroller has not been demonstrated and remains an embedded-validation objective.

\begin{figure*}[pos=!t]
    \centering
    \includegraphics[width=0.8\textwidth]{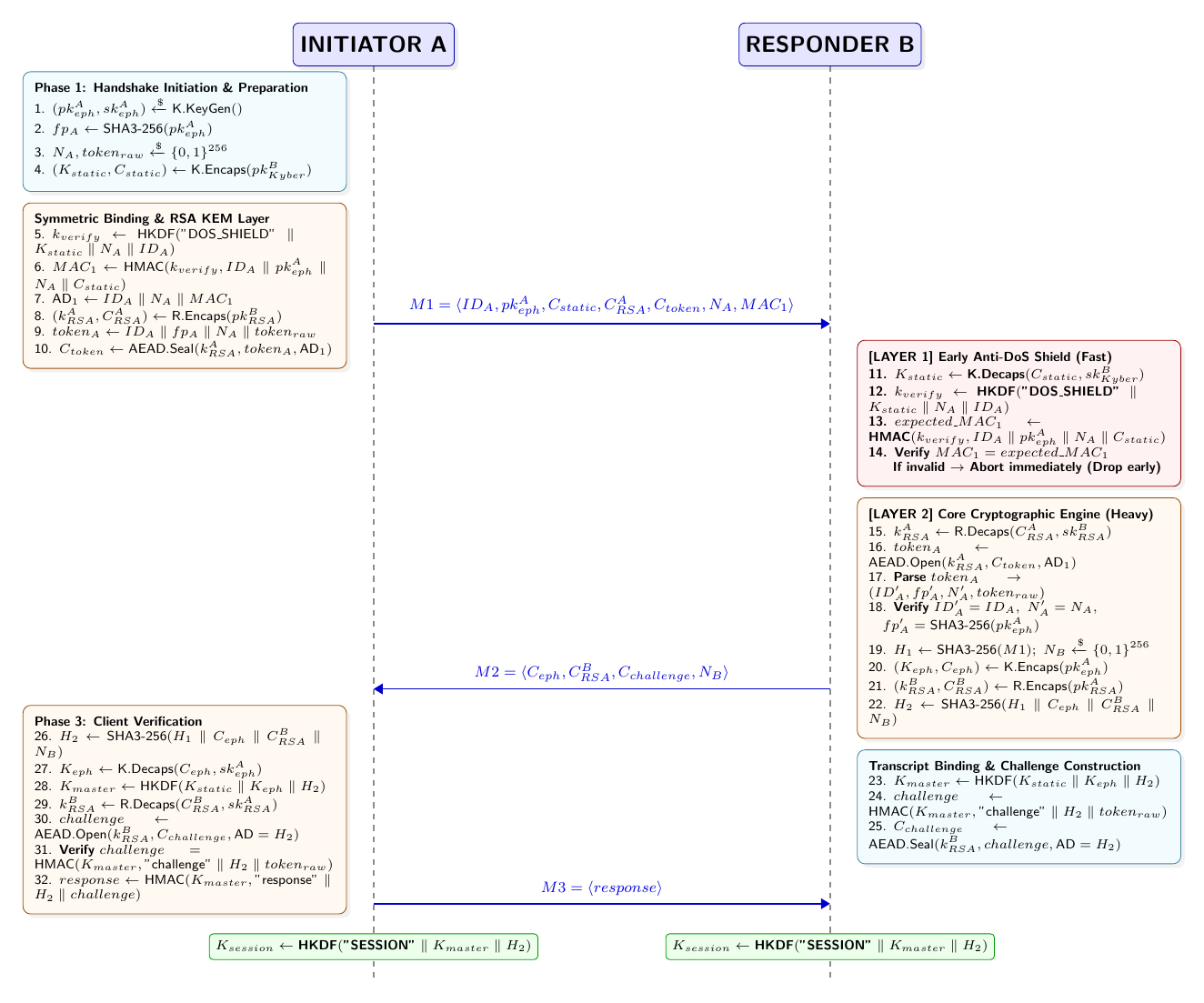}
    \caption{Handshake execution flow of HPQ-AKE.}
    \label{fig:protocol_flow_large}
\end{figure*}

\subsection{Notations and Primitives}
Table \ref{tab:notation} summarizes the cryptographic notations. To eliminate the runtime overhead of transmitting and validating heavy certificate chains in bandwidth-constrained channels, HPQ-AKE adopts a \textit{pre-cached provisioning model}. In this architectural setup, the long-term static public keys ($pk_{RSA}^A, pk_{RSA}^B, pk_{Kyber}^B$) are securely pre-loaded onto the respective nodes (e.g., during factory commissioning or initial device deployment). This design avoids the need for an active, online Public Key Infrastructure (PKI) lookup during the handshake. The authenticity of the pre-cached public keys is assumed to be established during the provisioning phase through an out-of-band trusted mechanism (e.g., manufacturer-installed certificates, secure enrollment, or enterprise PKI provisioning). Therefore, HPQ-AKE focuses on reducing online handshake overhead rather than replacing the trust establishment process itself.

\noindent\textbf{\rev{Key-update workflow (deployment guidance).}}
\rev{For scheduled rotation, an endpoint generates a replacement key pair, and an administrator authorizes an update binding the identity, key role, new public key, version, and validity interval under a separately provisioned management trust anchor. Recipients verify authorization and validity using trustworthy time, reject rollback using persistently stored versions, and acknowledge staged installation over an authenticated management channel. The administrator coordinates activation between handshakes and retirement of old credentials; unsynchronized peers must resynchronize rather than fall back to retired keys. On notification of suspected compromise, affected credentials should be disabled and replaced through an uncompromised management anchor or out-of-band commissioning, not authorization based solely on the compromised key. Such management operations may use signatures without adding transcript signatures to HPQ-AKE's three-pass handshake.}

\noindent\textbf{\rev{Trade-offs versus online PKI.}}
\rev{Pre-caching reduces per-handshake credential transmission and dependence on reachable credential services, but shifts distribution, storage, rotation, and recovery costs to deployment management. It remains compatible with PKI-based enrollment: certificate-path validation itself need not contact a CA for every handshake, while revocation information requires appropriate freshness checks \mbox{\cite{rfc5280KeyUsage}}. Reachable online status services can shorten revocation-notification delays, at the cost of communication and availability dependencies. Disconnected nodes may continue accepting a revoked key until notified or its local validity expires; a conservative policy rejects new handshakes when credential validity or management-state freshness expires until an authenticated refresh is obtained. Thus, shorter validity or refresh intervals reduce stale-key exposure but increase update traffic and the risk of service interruption. These workflows are deployment guidance, not an implemented or formally verified key-management extension; the present proof assumes authentic provisioned keys, and reported handshake costs exclude enrollment, updates, and recovery.}

\noindent\textbf{\rev{Fleet-scale limitations.}}
\rev{At fleet scale, authentic provisioning must be maintained across gateway/service endpoints and their authorized peer relationships; it is not guaranteed by the handshake. Coordinated rotation, version tracking, and catch-up for offline peers require an authenticated management service whose storage and update traffic depend on peer counts and churn. Compromise of the provisioning trust anchor can invalidate the assumed identity-to-key bindings. Certificate-based enrollment can supply these bindings, but transmitting certificates during handshakes would additionally require accounting for certificate chains, validation, and revocation-status handling. These deployment-specific costs are not included in the reported savings; the Full TLS baseline does not quantify a certificate-carrying HPQ-AKE variant.}

\begin{table*}[pos=htbp]
\centering
\caption{Cryptographic Notations and Primitives}
\label{tab:notation}
\renewcommand{\arraystretch}{1.2}
\begin{tabularx}{\textwidth}{l X}
\toprule
\textbf{Symbol} & \textbf{Description} \\
\midrule
$ID_A, ID_B$ & Unique identifiers for Initiator A and Responder B \\
$pk_{Kyber}^{B}, sk_{Kyber}^{B}$ & Static ML-KEM-768 key pair of Responder B (Pre-cached at Initiator A) \\
$(pk, sk)_{RSA}^{User}$ & Static RSA-3072 long-term key pairs for entities (Pre-loaded securely) \\
$(pk, sk)_{eph}^A$ & Ephemeral ML-KEM key pair generated by A \\
\texttt{K.Encaps/Decaps} & Kyber (ML-KEM) encapsulation and decapsulation \\
\texttt{R.Encaps/Decaps} & RSA-OAEP encapsulation and decapsulation \\
$(C_{static}, K_{static})$ & Static Kyber ciphertext and shared secret \\
$(C_{eph}, K_{eph})$ & Ephemeral Kyber ciphertext and shared secret \\
$(C_{RSA}^X, k_{RSA}^X)$ & RSA-KEM ciphertext and key generated by entity X \\
$fp_A$ & Fingerprint: $fp_A = \text{SHA3-256}(\text{encode}(pk_{eph}^A))$ \\
$K_{master}, K_{session}$ & Master secret and final Session Key \\
$\Challenge$ & HMAC challenge generated by B and bound to $K_{master}$, the transcript hash, and $token_{raw}$ for explicit key confirmation \\
$\Response$ & HMAC response generated by A over the verified challenge to prove possession of $K_{master}$ \\
$\ChallengeLabel, \ResponseLabel$ & Constant HMAC domain-separation labels \texttt{"challenge"} and \texttt{"response"}; only the labels are fixed, not the resulting values \\
$\mathsf{HKDF}(\cdot)$ & HMAC-based Key Derivation Function (RFC 5869) \\
\texttt{SecureErase} & Best-effort memory sanitization of temporary secrets after session establishment \\
\bottomrule
\end{tabularx}
\end{table*}

\subsection{Protocol Execution Details}

\noindent \textbf{Phase 1: Initiation (A $\rightarrow$ B).} 
Initiator A generates an ephemeral post-quantum key pair and encapsulates two cryptographic secrets: one targeted at B's static quantum key ($K_{static}$) and another targeted at B's long-term RSA public key ($k_{RSA}^A$). To safeguard identity privacy, a payload token binding A's identity, nonces, and the ephemeral key fingerprint ($fp_A$) is encrypted via authenticated encryption (AEAD) using $k_{RSA}^A$. Additionally, an early asymmetric DoS mitigation tag ($MAC_1$) is computed using a key derived from $K_{static}$. The resulting handshake initiation message is transmitted as
\[
 M1 = \langle ID_A, pk_{eph}^A, C_{static}, C_{RSA}^A, C_{token}, N_A, MAC_1 \rangle.
\]
\textit{Security Note:} \rev{These checks validate message consistency and token bindings, not yet the identity of A: anyone can encapsulate to public keys. Authentication of A is deferred until B verifies the final confirmation under a master key incorporating the RSA secret encapsulated to A.}

\noindent \textbf{Phase 2: Authentication \& Challenge (B $\rightarrow$ A).} 
Upon receiving $M1$, Responder B executes an accelerated verification phase by decapsulating $C_{static}$ to recompute and validate $MAC_1$; messages with mismatched tags are immediately discarded prior to invoking expensive RSA operations. Once the DoS filter is passed, B decapsulates $C_{RSA}^A$, decrypts $C_{token}$, and verifies that the internal identity and fingerprint bindings strictly match the outer claims. B then generates a fresh nonce $N_B$, encapsulates an ephemeral post-quantum secret ($K_{eph}$) against A's ephemeral key, and encapsulates an RSA secret ($k_{RSA}^B$) against A's long-term public key $pk_{RSA}^A$. A cryptographic challenge is derived from the master secret $K_{master}$ and sealed via AEAD under $k_{RSA}^B$ as follow
\[
 M2 = \langle C_{eph}, C_{RSA}^B, C_{challenge}, N_B \rangle.
\]
\rev{The master-key schedule explicitly incorporates the existing RSA secret generated by B:}
\[
\mbox{$K_{master} = \MasterKDF(K_{static} \parallel K_{eph} \parallel k_{RSA}^B \parallel H_2).$}
\]
\rev{Here the subscript denotes a separate HKDF context, with the ASCII string HPQ-AKE/master as its expand-stage info; the concatenated fields are its input keying material. This binds both confirmation tags and the session key to the RSA encapsulation to A, rather than relying solely on encryption of a challenge that could otherwise be recomputed. No handshake field or additional encapsulation is introduced.}

\noindent \textbf{Phase 3: Confirmation (A $\rightarrow$ B).} 
\rev{Initiator A receives M2 and first decapsulates both the ephemeral ML-KEM ciphertext and the RSA ciphertext addressed to A. Only after both operations succeed does it derive the master key and open the encrypted challenge using the recovered RSA secret.} After verifying that the recovered challenge value $\Challenge$ matches the expected transcript-bound value, A computes a confirmation response value $\Response$ utilizing an HMAC-based PRF layer to explicitly prove live session freshness and key possession
\[
 M3 = \langle \Response \rangle.
\]
The responder validates $\Response$ against its local state; upon a successful match, both parties compute the final keys $K_{session}$ and securely erase all ephemeral intermediates.

\subsection{Algorithmic Specification}
Algorithms \ref{alg:hardened_hpq_ake_initiator} and \ref{alg:hardened_hpq_ake_responder} detail the exact execution flow of the HPQ-AKE protocol. \rev{All Seal/Open calls use the nonce constants and single-encryption-per-key discipline specified in Section 3.3; retransmissions send cached messages without re-entering Seal.}

\begin{algorithm*}[t]
\caption{Hardened HPQ-AKE Initiator (A)}
\label{alg:hardened_hpq_ake_initiator}
\footnotesize
\begin{algorithmic}[1]
\renewcommand{\algorithmicrequire}{\textbf{Input:}}
\renewcommand{\algorithmicensure}{\textbf{Output:}}
\REQUIRE Long-term keys ($pk_{RSA}^B, pk_{Kyber}^B, sk_{RSA}^A$), Identities ($ID_A, ID_B$)
\ENSURE Final Session Key $K_{session}$ or termination token $\perp$

\STATE \textbf{// Phase 1: prepare M1}
\STATE $(pk_{eph}^A, sk_{eph}^A) \xleftarrow{\$} \text{K.KeyGen}()$
\STATE $fp_A \gets \text{SHA3-256}(pk_{eph}^A)$ \textbf{ // fingerprint}
\STATE $N_A \xleftarrow{\$} \{0,1\}^{256}; \quad token_{raw} \xleftarrow{\$} \{0,1\}^{256}$

\STATE \textbf{// 1. Early Asymmetric DoS Mitigation Filter (Fast Kyber Encapsulation)}
\STATE $(K_{static}, C_{static}) \gets \text{K.Encaps}(pk_{Kyber}^B)$
\STATE $k_{verify} \gets \text{HKDF}(\text{"DOS\_SHIELD"} \parallel K_{static} \parallel N_A \parallel ID_A)$
\STATE $MAC_1 \gets \text{HMAC}(k_{verify}, ID_A \parallel pk_{eph}^A \parallel N_A \parallel C_{static})$

\STATE \textbf{// 2. RSA core}
\STATE $(k_{RSA}^A, C_{RSA}^A) \gets \text{R.Encaps}(pk_{RSA}^B)$
\STATE $token_A \gets ID_A \parallel fp_A \parallel N_A \parallel token_{raw}$
\STATE $C_{token} \gets \text{AEAD.Seal}(k_{RSA}^A, token_A, \mbox{$N=\nu_T$}, \text{AD}=ID_A \parallel N_A \parallel MAC_1)$

\STATE \textbf{// 3. Build M1}
\STATE $M1 \gets \langle ID_A, pk_{eph}^A, C_{static}, C_{RSA}^A, C_{token}, N_A, MAC_1 \rangle$
\STATE $H_1 \gets \text{SHA3-256}(M1)$ \textbf{ // bind M1}
\STATE \textbf{Send} $M1$ to B

\STATE \textbf{// Phase 3: verify M2 and finish}
\STATE \textbf{Receive} $M2 = \langle C_{eph}, C_{RSA}^B, C_{challenge}, N_B \rangle$ from B
\STATE \textbf{// Transcript hash}
\STATE $H_2 \gets \text{SHA3-256}(H_1 \parallel C_{eph} \parallel C_{RSA}^B \parallel N_B)$

\STATE $K_{eph} \gets \text{K.Decaps}(C_{eph}, sk_{eph}^A)$
\IF{$K_{eph} = \perp$}
    \STATE \RETURN $\perp$ \textbf{ (Abort)}
\ENDIF

\STATE \mbox{$k_{RSA}^B \gets \text{R.Decaps}(C_{RSA}^B, sk_{RSA}^A)$}
\IF{\mbox{$k_{RSA}^B = \perp$}}
    \STATE \rev{Return $\perp$ (Abort)}
\ENDIF
\STATE \mbox{$K_{master} \gets \MasterKDF(K_{static} \parallel K_{eph} \parallel k_{RSA}^B \parallel H_2)$}

\STATE $\Challenge \gets \text{AEAD.Open}(k_{RSA}^B, C_{challenge}, \mbox{$N=\nu_C$}, \text{AD}=H_2)$
\IF{$\Challenge = \perp$}
    \STATE \RETURN $\perp$ \textbf{ (Abort)}
\ENDIF

\STATE $expected \gets \text{HMAC}(K_{master}, \ChallengeLabel \parallel H_2 \parallel token_{raw})$
\IF{$\Challenge \neq expected$}
    \STATE \RETURN $\perp$ \textbf{ (Abort)}
\ENDIF

\STATE $\Response \gets \text{HMAC}(K_{master}, \ResponseLabel \parallel H_2 \parallel \Challenge)$
\STATE \textbf{Send} $M3 = \langle \Response \rangle$ to B

\STATE $K_{session} \gets \text{HKDF}(\text{"SESSION"} \parallel K_{master} \parallel H_2)$
\STATE \text{SecureErase}(\text{all temporary secrets})
\STATE \RETURN $K_{session}$
\end{algorithmic}
\end{algorithm*}

\begin{algorithm*}[t]
\caption{Hardened HPQ-AKE Responder (B)}
\label{alg:hardened_hpq_ake_responder}
\footnotesize
\begin{algorithmic}[1]
\renewcommand{\algorithmicrequire}{\textbf{Input:}}
\renewcommand{\algorithmicensure}{\textbf{Output:}}
\REQUIRE Long-term keys ($sk_{RSA}^B, sk_{Kyber}^B, pk_{RSA}^A$), Identities ($ID_A, ID_B$)
\ENSURE Final Session Key $K_{session}$ or termination token $\perp$

\STATE \textbf{// Phase 2: check M1 and build challenge}
\STATE \textbf{Receive} $M1 = \langle ID_A, pk_{eph}^A, C_{static}, C_{RSA}^A, C_{token}, N_A, MAC_1 \rangle$ from A

\STATE \textbf{// STEP 2.1: ASYMMETRIC DOS MITIGATION (Ultra-fast Kyber Decaps Verification)}
\STATE $K_{static} \gets \text{K.Decaps}(C_{static}, sk_{Kyber}^B)$
\IF{$K_{static} = \perp$}
    \STATE \RETURN $\perp$ \textbf{ (Abort)}
\ENDIF
\STATE $k_{verify} \gets \text{HKDF}(\text{"DOS\_SHIELD"} \parallel K_{static} \parallel N_A \parallel ID_A)$
\STATE $expected\_MAC_1 \gets \text{HMAC}(k_{verify}, ID_A \parallel pk_{eph}^A \parallel N_A \parallel C_{static})$
\IF{$MAC_1 \neq expected\_MAC_1$}
    \STATE \RETURN $\perp$ \textbf{ (Abort - Invalid payload dropped prior to RSA decryption)}
\ENDIF

\STATE \textbf{// 2.2 RSA decapsulation}
\STATE $k_{RSA}^A \gets \text{R.Decaps}(C_{RSA}^A, sk_{RSA}^B)$
\IF{$k_{RSA}^A = \perp$}
    \STATE \RETURN $\perp$ \textbf{ (Abort)}
\ENDIF
\STATE $token_A \gets \text{AEAD.Open}(k_{RSA}^A, C_{token}, \mbox{$N=\nu_T$}, \text{AD}=ID_A \parallel N_A \parallel MAC_1)$
\IF{$token_A = \perp$}
    \STATE \RETURN $\perp$ \textbf{ (Abort)}
\ENDIF

\STATE Parse $token_A \to (ID'_A, fp'_A, N'_A, token_{raw})$
\STATE \textbf{// Check identity, nonce, fingerprint}
\IF{$ID'_A \neq ID_A$ \textbf{or} $N'_A \neq N_A$ \textbf{or} $fp'_A \neq \text{SHA3-256}(pk_{eph}^A)$}
    \STATE \RETURN $\perp$ \textbf{ (Abort)}
\ENDIF

\STATE \textbf{// 2.3 Build M2}
\STATE $H_1 \gets \text{SHA3-256}(M1)$
\STATE $N_B \xleftarrow{\$} \{0,1\}^{256}$
\STATE $(K_{eph}, C_{eph}) \gets \text{K.Encaps}(pk_{eph}^A)$
\STATE $(k_{RSA}^B, C_{RSA}^B) \gets \text{R.Encaps}(pk_{RSA}^A)$

\STATE $H_2 \gets \text{SHA3-256}(H_1 \parallel C_{eph} \parallel C_{RSA}^B \parallel N_B)$ \textbf{ // bind}
\STATE \mbox{$K_{master} \gets \MasterKDF(K_{static} \parallel K_{eph} \parallel k_{RSA}^B \parallel H_2)$}

\STATE $\Challenge \gets \text{HMAC}(K_{master}, \ChallengeLabel \parallel H_2 \parallel token_{raw})$
\STATE $C_{challenge} \gets \text{AEAD.Seal}(k_{RSA}^B, \Challenge, \mbox{$N=\nu_C$}, \text{AD}=H_2)$
\STATE \textbf{Send} $M2 = \langle C_{eph}, C_{RSA}^B, C_{challenge}, N_B \rangle$ to A

\STATE \textbf{// Phase 3: verify M3 and finish}
\STATE \textbf{Receive} $M3 = \langle \Response \rangle$ from A
\STATE $expected \gets \text{HMAC}(K_{master}, \ResponseLabel \parallel H_2 \parallel \Challenge)$
\IF{$\Response \neq expected$}
    \STATE \RETURN $\perp$ \textbf{ (Abort)}
\ENDIF

\STATE $K_{session} \gets \text{HKDF}(\text{"SESSION"} \parallel K_{master} \parallel H_2)$
\STATE \text{SecureErase}(\text{all temporary secrets})
\STATE \RETURN $K_{session}$
\end{algorithmic}
\end{algorithm*}

\section{Security Analysis}
\label{sec:security_analysis}
\subsection{Threat Model and Security Objectives}
HPQ-AKE is designed to provide authenticated key exchange security against a classical-quantum hybrid adversary during the post-quantum transition period. We assume the adversary possesses full classical network control (can intercept, replay, modify, and inject messages)
and bounded quantum computational resources insufficient to run large-scale Shor's algorithm but capable of quantum search on symmetric keys and quantum random oracle queries against lattice-based constructions.
The protocol targets the following security objectives, formally
established in Sections 5.2–5.6:
\begin{enumerate}
\item Authenticated Key Establishment: Sessions derive identical session keys only between honest matching partners.
\item Perfect Forward Secrecy (PFS): Compromise of long-term static keys does not expose past session keys.
\item \rev{Key Compromise Impersonation (KCI) Resistance: Resistance to impersonating an uncompromised peer after disclosure of the victim\textquotesingle s long-term keys, under the restricted exposure and KDF assumptions in Proposition~\mbox{\ref{prop:kci}}.}
\item Unknown Key Share (UKS) Resistance: Session identifiers cryptographically bind communication partners, preventing misbinding attacks.
\item Replay Resistance: Fresh nonces and session transcript binding prevent session state replay.
\item Deniability: KEM-based implicit authentication provides off-the-record security as a bonus property.
\end{enumerate}

We analyze the security of HPQ-AKE in an extended Bellare--Rogaway authenticated key exchange (AKE) framework \cite{b1}, augmented with forward secrecy and adversarial capabilities inspired by the eCK model \cite{b2} and the Canetti--Krawczyk (CK) framework \cite{Canetti2001}. To accurately capture the transitional nature of our architecture, the formal reduction is established using a \textit{Hybrid Oracle Model Framework}: the integrity of the real-time active entity authentication is proven in the classical Random Oracle Model (ROM), whereas the confidentiality and forward secrecy of the session key against future quantum adversaries are rigorously bounded in the Quantum Random Oracle Model (QROM).

In the multi-user/multi-session setting, we let $n_p$ denote the total number of parties in the system and $n_s$ denote the maximum number of concurrent sessions initiated by each party. All security bounds are scaled by the multi-session factor $n_p\cdot n_s$ to account for an adversary targeting any of the $n_p \cdot n_s$ possible concurrent sessions running across the network topology \cite{ShoupGames, kemtls}.

\subsection{Formal Model and Correctness}
\textbf{Matching Conversations and Correctness.}
We first define the criteria for a successful honest execution.

\begin{definition}[Matching Conversations]
Two sessions $\Pi_A^i$ and $\Pi_B^j$ are said to have {\rm{matching conversations}} if the sequence of messages sent and received by A in $\Pi_A^i$ is identical to the sequence received and sent by B in $\Pi_B^j$, respectively. Formally
$$(M1, M2, M3)_{\Pi_A^i} = (M1, M2, M3)_{\Pi_B^j}.$$
\end{definition}

\begin{definition}[Correctness]
A protocol is {\rm{correct}} if, whenever two uncorrupted parties engage in matching conversations, both accept and derive the same session key
$$K_A^{\mathrm{session}} = K_B^{\mathrm{session}}.$$
\end{definition}

\begin{proposition}[Correctness of HPQ-AKE]
If A and B complete a protocol run with matching conversations, then
$$K_A^{\mathrm{master}} = K_B^{\mathrm{master}}, \quad K_A^{\mathrm{session}} = K_B^{\mathrm{session}}.$$
\end{proposition}

\begin{proof}
\rev{Correctness of ML-KEM and RSA-OAEP gives both parties the same static ML-KEM secret, ephemeral ML-KEM secret, and RSA secret encapsulated by B to A, except for the underlying decapsulation failure probability. Matching transcripts give the same master-KDF inputs in Section 4.3. Determinism of HKDF and HMAC then gives identical master keys, confirmation values, and session keys. This correctness argument does not require treating encryption to a public key as evidence of sender identity.}
\end{proof}

\subsection*{Partnering and Session Freshness}
A session is identified by its local execution transcript $(M1, M2, M3)$, and the session identifier is uniquely defined as:
\[
SID = H(M1 \parallel M2 \parallel M3).
\]
Two sessions $\Pi_A^i$ and $\Pi_B^j$ are \emph{partners} if they both accept and hold the identical $SID$ \cite{b1}. To explicitly prevent parsing ambiguities, structural alignment shifts, and cryptographic canonicalization attacks during transcript binding, all discrete constituent message framing fields, nonces, and cryptographic tokens within $(M1, M2, M3)$ are strictly length-prefixed (utilizing fixed 16-bit length-encoding metadata descriptors) prior to any string concatenation operation ($\parallel$).

A session $\Pi$ is defined as \emph{fresh} if:
\begin{enumerate}
    \item Neither $\Pi$ nor its designated partner has had its session key revealed;
    \item The adversary has not simultaneously obtained both (i) all long-term static secret keys of one party and (ii) the ephemeral secret key of the peer for this specific session;
    \item A unique matching partner exists (to exclude trivial active impersonation scenarios);
    \item The static secret keys of both communication parties remain uncompromised until the moment of session acceptance.
\end{enumerate}
\rev{These are the freshness restrictions for the Test-session secrecy experiment, not a full eCK freshness definition. In particular, they exclude pre-acceptance corruption and cannot themselves establish KCI resistance. Proposition~\mbox{\ref{prop:kci}} therefore states a separate impersonation experiment allowing victim-key disclosure and acceptance without a matching peer.}

\subsection*{Oracles Available to the Adversary}
The active adversary interacts with the execution environment through the following standard oracles \cite{b1, b2}:
\begin{itemize}
    \item \textbf{Send($\Pi$, msg)}: Delivers an arbitrary network message to a targeted session, returning its generated response.
    \item \textbf{RevealSessionKey($\Pi$)}: Returns the derived and accepted session key of an active session.
    \item {\RaggedRight\emergencystretch=3em \textbf{RevealEphemeral($\Pi$)}: Discloses the temporary ephemeral secrets of a targeted session.\par}
    \item \textbf{Corrupt($U$)}: Returns the long-term static secret keys of a designated party $U$.
    \item \textbf{Test($\Pi$)}: For a strictly fresh session $\Pi$, returns either the real session key or a uniform random string of equal length based on a hidden internal bit $b \in \{0,1\}$.
\end{itemize}

\subsection*{Session Key Indistinguishability Game}
The Test oracle formally evaluates the cryptographic strength of the protocol:
\begin{itemize}
    \item For a fresh session $\Pi_U^i$, Test returns $K_{\mathrm{session}}$ (if $b=1$) or a uniform random string of equal length (if $b=0$).
    \item The adversary evaluates the output and returns a guess bit $b'$.
    \item The final indistinguishability advantage is defined as:
    \[
    \mathsf{Adv}^{\mathsf{HPQ\text{-}AKE}}_{\mathcal{A}} = \left|\, \Pr[b' = b] - \tfrac{1}{2} \,\right|
    \]
\end{itemize}

\subsection{Secrecy Analysis under Hybrid Oracle Models}
\label{sec:hybrid_secrecy}

To accurately reflect the transitional threat model of HPQ-AKE, we divide our formal analysis. The integrity of the real-time active entity authentication is proven in the classical Random Oracle Model (ROM) based on the RSA assumption, while the confidentiality and forward secrecy of the derived session key against future quantum adversaries are bounded in the Quantum Random Oracle Model (QROM) based on the Module-LWE assumption.

\begin{theorem}[Hybrid Session Key Indistinguishability]
\label{thm:hybrid_main}
Let $\mathcal{A}$ be a Probabilistic Polynomial-Time (PPT) adversary attacking the HPQ-AKE protocol. In a multi-user, multi-session environment with $n_p$ parties and $n_s$ sessions, the advantage of $\mathcal{A}$ in the Session Key Indistinguishability game is bounded by
\begin{equation}
    \mathsf{Adv}_{\mathrm{HPQ\text{-}AKE}}^{\mathrm{IND\text{-}AKE}}(\mathcal{A}) \le \epsilon_{\mathrm{ROM}}^{\mathrm{Auth}} + \epsilon_{\mathrm{QROM}}^{\mathrm{Secrecy}}.
\end{equation}
where the classical authentication failure bound $\epsilon_{\mathrm{ROM}}^{\mathrm{Auth}}$ covering the traditional cryptographic primitives is defined as
\begin{equation}
\begin{aligned}
\varepsilon^{\mathrm{Auth}}_{\mathrm{ROM}}
&= n_p n_s \Bigl(\mathsf{Adv}_{\mathrm{ML\text{-}KEM, static}}^{\mathrm{IND\text{-}CCA2}}(\mathcal{A})
+ \mathsf{Adv}_{\mathrm{RSA}}^{\mathrm{IND\text{-}CCA}}(\mathcal{A}) \\
&\qquad + \frac{q_{\mathrm{AEAD}}}{2^{128}} + \frac{q_{\mathrm{PRF}}^2}{2^{256}}\Bigr).
\end{aligned}
\end{equation}

and the post-quantum secrecy bound $\epsilon_{\mathrm{QROM}}^{\mathrm{Secrecy}}$ against quantum oracle queries bounded by $q_H$ is defined as
\begin{equation}
    \epsilon_{\mathrm{QROM}}^{\mathrm{Secrecy}} = n_p n_s \left( \mathsf{Adv}_{\mathrm{ML\text{-}KEM, eph}}^{\mathrm{IND\text{-}CCA2}}(\mathcal{A}_Q) + \mathcal{O}\left(\frac{q_H^3}{2^{256}}\right) \right).
\end{equation}
\end{theorem}
Note that the ephemeral ML-KEM advantage term is counted only in the QROM secrecy bound. The authentication bound captures the classical mechanisms that enforce entity binding and transcript validity (static ML-KEM filtering, RSA-OAEP implicit authentication, AEAD integrity, and HMAC/HKDF pseudorandomness), while the ephemeral ML-KEM transition in Game $G_5$ accounts for forward secrecy and session-key indistinguishability. This separation avoids double-counting the same hybrid-game transition in the total advantage bound.

\textbf{Concrete Security and Tightness:}
\rev{The multi-session factor affects the concrete interpretation of the bounds. RSA-3072 provides only a classical transitional assumption, not post-quantum security or a guarantee that reduction losses are negligible in practice. The KCI claim below excludes adversaries capable of breaking the uncompromised peer\textquotesingle s RSA key; ML-KEM-based secrecy must not be conflated with quantum-resistant RSA authentication.}

\rev{The following Test-session games do not constitute a KCI reduction. The role-specific authentication argument for the revised key schedule, including its additional ideal-KDF assumption, is stated separately in Proposition~\mbox{\ref{prop:kci}}.}

\subsubsection{Proof via a Highly Granular Sequence of Hybrid Games}
We establish the reduction using a strict sequence of hybrid games $G_0$ to $G_6$ to formally isolate every cryptographic transition.

\paragraph{\textbf{Game $G_0$}.} This represents the real execution of the HPQ-AKE protocol under the control of the active adversary $\mathcal{A}$. By definition
$$\mathsf{Adv}_{G_0} = \mathsf{Adv}_{\mathrm{HPQ\text{-}AKE}}^{\mathrm{IND\text{-}AKE}}(\mathcal{A}).$$

\paragraph{\textbf{Game $G_1$ (Static KEM Mitigation Filter Transition)}.} 
In this game, we transition the static ML-KEM encapsulation layer ($(K_{\text{static}}, C_{\text{static}})$ in lines 6--7 of Algorithm 1) used for early asymmetric DoS filtering into an ideal IND-CCA2 simulator. Consistent with the Fujisaki-Okamoto (FO) transform design within ML-KEM, the simulator enforces an \textit{implicit rejection} mechanism. If the adversary manipulates $C_{\text{static}}$ to probe the filter, the decapsulation oracle returns a deterministic pseudorandom string rather than a structural error token $\bot$, suppressing any decryption oracle behavior. This transition accounts for the static KEM security margin
$$| \mathsf{Adv}_{G_1} - \mathsf{Adv}_{G_0} | \le n_p n_s \cdot \mathsf{Adv}_{\mathrm{ML\text{-}KEM, static}}^{\mathrm{IND\text{-}CCA2}}(\mathcal{A}).$$

\paragraph{\textbf{Game $G_2$ (RSA-OAEP Core Decoupling)}.} 
In this game, we replace the legacy RSA key encapsulation core (implemented via RSA-OAEP in line 10 of Algorithm 1) with an ideal IND-CCA simulator. The simulator embeds a target RSA challenge $(pk_{\text{RSA}}^*, C^*)$ into message $M1$ and proxies any alternative active decryption queries to an external classical oracle. Because we assume the adversary executing real-time active man-in-the-middle impersonations possesses only classical capabilities in the present deployment phase (pre-CRQC), this transition incurs standard classical bounds
$$| \mathsf{Adv}_{G_2} - \mathsf{Adv}_{G_1} | \le n_p n_s \cdot \mathsf{Adv}_{\mathrm{RSA}}^{\mathrm{IND\text{-}CCA}}(\mathcal{A}).$$

\paragraph{\textbf{Game $G_3$ (Symmetric AEAD Layer Integrity)}.}
\rev{The AEAD integrity assumption here requires the nonce-respecting key-use discipline in Section 3.3; it does not cover repeated key--nonce pairs.} 
We now replace the symmetric keys derived from the RSA core to protect the handshake payloads $C_{\text{token}}$ with ideal AEAD simulators. Under the IND-CCA and ciphertext integrity (INT-CTXT) properties of AES-GCM, any adversarial attempt to manipulate, alter, or forge the handshake packets without the legitimate key will trigger an immediate session abort (lines 21--25 of Algorithm 2). The distinguishing advantage introduced is bounded by
$$| \mathsf{Adv}_{G_3} - \mathsf{Adv}_{G_2} | \le n_p n_s \cdot \frac{q_{\mathrm{AEAD}}}{2^{128}}.$$

\paragraph{\textbf{Game $G_4$ (HKDF / HMAC Pseudorandomness)}.} 
\rev{This step uses an ideal-KDF model for the master-key derivation, together with idealized HMAC authentication tags. The displayed birthday term accounts for collisions in the ideal-function model; it is not a standalone reduction establishing the extractor security of an arbitrary HKDF instantiation from HMAC PRF security alone. The additional KCI argument explicitly conditions on this KDF idealization in Proposition~\mbox{\ref{prop:kci}}.}
$$| \mathsf{Adv}_{G_4} - \mathsf{Adv}_{G_3} | \le n_p n_s \cdot \frac{q_{\mathrm{PRF}}^2}{2^{256}}.$$

\paragraph{\textbf{Game $G_5$ (Ephemeral KEM Layer in QROM)}.} 
To guarantee perfect forward secrecy (PFS) against future quantum adversaries capable of storing historical classical traffic for retroactive decryption, we transition the ephemeral ML-KEM encapsulation layer ($(C_{eph}, K_{eph})$ in line 35 of Algorithm 2). Since the IND-CCA2 security of ML-KEM already holds robustly in the Quantum Random Oracle Model (QROM) as formalized in FIPS 203 \cite{nist-fips203-Alagic}, the simulator can directly invoke the ML-KEM decapsulation oracle to resolve quantum queries without explicit random oracle reprogramming. By applying a standard hybrid argument over $n_p$ registered identities and $n_s$ active sessions to guess the target session, the advantage gap between $G_4$ and $G_5$ is tightly bounded by the post-quantum IND-CCA2 security of the ephemeral KEM component
$$| \mathsf{Adv}_{G_5} - \mathsf{Adv}_{G_4} | \le n_p n_s \cdot \mathsf{Adv}_{\mathrm{ML\text{-}KEM, eph}}^{\mathrm{IND\text{-}CCA2}}(\mathcal{A}_Q).$$

\paragraph{\textbf{Game $G_6$ (Quantum Public Key Fingerprint Collisions)}.} 
The protocol tightly binds the ephemeral public key $pk_{\text{eph}}^A$ to the cryptographic token via the truncated fingerprint $fp_A = \mathsf{SHA3\text{-}256}(pk_{\text{eph}}^A)$ (line 3 of Algorithm 1). In this final game, we rule out quantum-level hash collisions. A quantum adversary executing optimal quantum collision-finding algorithms (such as the BHT algorithm) could attempt to locate a collision for $pk_{\text{eph}}^A$ that yields the identical $fp_A$, executing a public-key substitution attack. Since HPQ-AKE enforces a strict 256-bit truncation for $fp_A$, the quantum collision advantage scales cubically according to generic bounds in the Quantum Random Oracle Model (QROM)
$$| \mathsf{Adv}_{G_6} - \mathsf{Adv}_{G_5} | \le n_p n_s \cdot \mathcal{O}\left(\frac{q_H^3}{2^{256}}\right),$$

where this cubic threshold against quantum query transitions is formally justified based on the foundational QROM reduction and collision bounding techniques established by Zhandry \cite{zhandry2015secure}.
\paragraph{Final Key Derivation Step} 
Conditioned on the exclusion of quantum fingerprint collisions in Game $G_6$, any adversarial manipulation of the public keys triggers an immediate authentication failure. Furthermore, as the HKDF output relies directly on the high-entropy pool of $K_{\text{eph}}$—which is now completely indistinguishable from random in the QROM—the final session key $K_{\mathrm{session}}$ is computationally independent of the protocol network transcript. Consequently, the adversary has no remaining data correlation, leading to a winning probability of exactly $1/2$. Collecting the telescoping sum across all game phases completes the formal proof of Theorem 1.

\subsection{Authentication and Robustness}
\subsubsection{Authentication Security}
The protocol enforces both implicit identity binding and explicit key confirmation.

\begin{definition}[Implicit Authentication]
A protocol provides {\rm{implicit authentication}} if, whenever a session $\Pi_A$ accepts with peer identity B, then any partnering session $\Pi_B$ must be the unique session capable of generating the cryptographic payloads accepted by $\Pi_A$, except with negligible probability \cite{b1}.
\end{definition}

\begin{theorem}
The proposed protocol achieves mutual implicit authentication within the hybrid oracle framework, bounded by
\begin{equation}
\begin{aligned}
\mathsf{Adv}^{\mathrm{auth}}(\mathcal{A})
&\le n_p n_s \Bigl(\mathsf{Adv}_{\mathrm{ML\text{-}KEM, static}}^{\mathrm{IND\text{-}CCA2}}(\mathcal{A}) \\
&\quad + \mathsf{Adv}_{\mathrm{RSA}}^{\mathrm{IND\text{-}CCA}}(\mathcal{A})
+ \frac{q_{\mathrm{AEAD}}}{2^{128}} + \frac{q_{\mathrm{PRF}}^2}{2^{256}} \\
&\quad + \mathcal{O}\left(\frac{q_H^3}{2^{256}}\right)\Bigr).
\end{aligned}
\end{equation}
\end{theorem}

\begin{proof}
Assume an active adversary attempts to forge an accepting handshake session at Initiator A without an active matching partner at Responder B. To successfully compromise the protocol layer, the adversary must execute one of the following: 
(i) break the IND-CCA boundary of the RSA layer, 
(ii) violate the INT-CTXT clause of the AEAD layer to manipulate token payloads ($q_{\mathrm{AEAD}}$ bounds), 
(iii) find a quantum hash collision in the 256-bit public-key fingerprint ($\mathcal{O}(q_H^3/2^{256})$ advantage), 
(iv) forge an HMAC tag without knowledge of $K_{\mathrm{master}}$ ($q_{\mathrm{PRF}}$ bounds),
(v) compromise the static identity ML-KEM layer ($\mathsf{Adv}_{\mathrm{ML\text{-}KEM, static}}^{\mathrm{IND\text{-}CCA2}}$). 
Since each primitive is computationally intractable under the established hybrid assumptions, the cumulative forgery advantage is bounded by the stated reduction scaled by the multi-session factor \cite{Avanzi, Bader2015, Gjosteen2018}.

\rev{For responder-side confirmation in the revised schedule, the relevant RSA secret is instead the one encapsulated by B to A in M2. Its inclusion in the master KDF prevents a chosen-M1 sender from completing authentication using only its chosen token and ML-KEM secrets. Victim-key disclosure is treated separately by the conditional argument in Proposition~\mbox{\ref{prop:kci}}; the preceding Test-session games are not asserted to prove that exposure case.}
\end{proof}

\begin{definition}[Explicit Authentication]
\rev{Explicit key confirmation provides evidence that the designated peer possesses the same key material for the matching transcript. In this three-message protocol, it does not imply simultaneous acceptance: A cannot know whether B has received M3.}
\end{definition}

\begin{corollary}[Explicit Authentication]
\rev{Under the authentication and key-derivation assumptions stated in Proposition~\mbox{\ref{prop:kci}}, the challenge and response provide role-specific explicit key confirmation.}
\end{corollary}
\begin{proof}
\rev{A verifies the challenge from B; B verifies the response from A. Both tags use the revised master key, distinct role labels, and the transcript hash. In particular, a sender choosing its own ephemeral key cannot compute B\textquotesingle s expected response merely from the two ML-KEM secrets and its own token: the master key also requires the RSA secret encapsulated to the claimed A. Successful confirmation is interpreted relative to a matching peer transcript prefix, not guaranteed delivery of the final message.}
\end{proof}

\subsubsection{Additional Security Properties}

\begin{proposition}[Forward Secrecy]
HPQ-AKE provides perfect forward secrecy (PFS) under the IND-CCA2 security of Kyber.
\end{proposition}
\begin{proof}
\rev{The revised master-key schedule in Section 4.3 retains the ephemeral ML-KEM secret as an input alongside the static ML-KEM and RSA secrets, under the stated KDF idealization.} Upon handshake completion, the ephemeral key $sk_{\text{eph}}^A$ is instantly purged via \texttt{SecureErase}. Consequently, a retrospective compromise of the long-term static assets ($sk_{\mathrm{RSA}}, sk_{\mathrm{Kyber}}$) grants a quantum adversary no advantage in recovering past keys, as recovering $K_{\text{eph}}$ requires solving the underlying Module-LWE problem on $C_{\text{eph}}$ \cite{nist-report2022}.
\end{proof}

\begin{proposition}[Key Compromise Impersonation Resistance]
\label{prop:kci}
\rev{For the revised key schedule, KCI resistance holds in the classical, long-term-key-only impersonation experiment below, conditional on the stated primitive and key-derivation assumptions. This is not a claim of unrestricted eCK security or post-quantum authentication.}
\end{proposition}

\noindent\textbf{\rev{Exposure and winning condition.}}
\rev{The adversary controls message delivery and may obtain A\textquotesingle s long-term RSA private key when A is the victim, or both of B\textquotesingle s long-term private keys when B is the victim, before the target run. The victim still executes honestly; all long-term keys of the claimed peer and the provisioned identity-to-key bindings remain uncompromised. No honest-session state, ephemeral secret, master key, or session key is revealed in this restricted experiment. Attacker-generated ephemeral keys and their secrets are allowed. A win is victim acceptance as the claimed peer without a matching honest peer transcript prefix: production of the corresponding M2 by B for an accepting A, or production of M3 after processing that M2 by A for an accepting B. Forwarding such an honest exchange is not impersonation, even if the last message is subsequently dropped.}

\noindent\textbf{\rev{Cryptographic presuppositions.}}
\rev{We assume classical IND-CCA security of the static ML-KEM and RSA-OAEP encapsulations, secure independent randomness, collision-resistant transcript/fingerprint hashing, authenticated identity-specific public keys, unambiguous field encoding, AEAD confidentiality and integrity, and HMAC PRF security. Additionally, the domain-separated master KDF is idealized as a random oracle: its output is hidden until the complete input is queried, including an unexposed KEM secret. This is an explicit modeling assumption, not a consequence of HMAC PRF security alone or a standalone proof for every HKDF instantiation.}

\begin{proof}
\rev{For victim A, disclosure of A\textquotesingle s RSA key does not disclose the static ML-KEM secret that honest A encapsulates to uncompromised B. An impersonator can create an ephemeral ciphertext and know its secret; that ability is not an IND-CCA break. The missing master-KDF input is instead the static ML-KEM secret, absent an allowed matching B execution or a violation of the assumptions above.}

\rev{For victim B, even disclosure of both B private keys leaves uncompromised A\textquotesingle s RSA private key unavailable. Honest B generates a fresh RSA secret encapsulated to A, and the revised master KDF includes it. Thus, an impersonator who chooses M1 and knows both ML-KEM secrets still lacks this RSA input. Encryption of a recomputable challenge alone would not suffice; its inclusion in the master key is essential.}

\rev{In each direction, condition on no nonce or transcript-binding collision. The relevant uncompromised-peer KEM protects the missing input, and the ideal-KDF assumption hides the master key unless that input is recovered or guessed. A non-matching accepting transcript then requires defeating this protection or forging the role-separated confirmation. Honest confirmation relays are covered by the matching-prefix condition. This gives the conditional KCI argument under the stated restricted exposure model; no tight multi-session KCI bound or machine-checked reduction is claimed.}
\end{proof}

\begin{proposition}[Unknown Key Share Resistance]
HPQ-AKE is secure against Unknown Key Share (UKS) attacks.
\end{proposition}
\begin{proof}
The Session Identifier ($SID$) and the KDF input vectors explicitly encapsulate the unique identities ($ID_A, ID_B$), session nonces ($N_A, N_B$), and the public key fingerprint $fp_A$. Any identity-misbinding or cross-session injection attempted by an adversary automatically triggers an AEAD authentication failure or an HMAC validation mismatch, forcing an immediate session abort \cite{Bader2015, rfc5869}.
\end{proof}

\subsection{Replay Resistance}
\begin{proposition}[Replay Resistance]
Replay attacks against the HPQ-AKE architecture are infeasible.
\end{proposition}
\begin{proof}
Suppose an adversary replays a recorded message:
$$M1 = \langle ID_A, pk_{\text{eph}}^A, C_{\text{static}}, C_{\mathrm{RSA}}^A, C_{\text{token}}, N_A, MAC_1 \rangle.$$
Upon receipt, the honest Responder B generates a fresh nonce $N_B \xleftarrow{\$} \{0,1\}^{256}$ and executes fresh encapsulation to produce a new $K_{\text{eph}}$. This ensures that $K_{\text{fresh}}^{\mathrm{master}} \neq K_{\text{replayed}}^{\mathrm{master}}$. Conversely, if an old message $M2$ is replayed against a new session, the fresh nonce $N_A$ generated by the Initiator ensures that the re-derived master key mismatches, causing the subsequent AEAD or HMAC validation checks to abort the handshake instantly \cite{b2}.
\end{proof}

\subsection{Authentication and Deniability}
Unlike traditional Signed-Diffie-Hellman handshakes that use digital signatures to enforce non-repudiation, HPQ-AKE employs a KEM-based sign-less implicit authentication mechanism. Because the handshake transcript is composed entirely of ciphertexts and symmetric MACs that either party could theoretically simulate or forge given the shared secrets, a third-party auditor cannot use a captured transcript as legally binding proof of participation. This grants HPQ-AKE the property of \textit{Off-the-Record Deniability}, a highly desired privacy feature in modern secure messaging and high-security tactical links. 

\rev{During active execution, authentication is conditional on the uncompromised-peer keys, trusted provisioning, and key-derivation assumptions stated above. Confirmation establishes possession of the transcript-bound master-key material; it is not an unconditional proof that both endpoints remain uncompromised or an assertion of equivalence to every certificate-based deployment.}

\section{Performance Evaluation}
\label{sec:performance}
\rev{Revision scope: the master-key correction in Section 4.3 adds the existing 32-byte RSA secret to local KDF input and separates the master-KDF context. It changes neither the three-message structure nor transmitted fields, output lengths, or the number of cryptographic operations. The original measured profiles, figures, and numerical model inputs below are retained; the corrected KDF has not been rebenchmarked, so these timings are not new measurements of the revised implementation.}
To rigorously evaluate the proposed \textbf{HPQ-AKE} protocol, we conducted a multi-layer comparative analysis against five representative baseline configurations:
\begin{itemize}
    \item \textbf{Kyber-only}: An unauthenticated Post-Quantum Key Encapsulation Mechanism (PQ-KEM) baseline serving as the absolute lower bound for computational and communication overhead.
    \item \textbf{TLS 1.3 (Classical)}: The conventional industry standard utilizing $X25519$ for key exchange and $ECDSA\text{-}P256$ for authentication.\par
    \item \textbf{Hybrid TLS 1.3 (Full Handshake)}: A NIST-compliant dual-layer protocol combining $ML\text{-}KEM\text{-}768$ and $ML\text{-}DSA\text{-}65$ to provide both post-quantum confidentiality and transcript signatures.\par
    \item \textbf{Hybrid TLS 1.3 (Pre-cached Baseline)}: An optimized operational state of the hybrid baseline where post-quantum certificate chains are omitted from the transmission wire, simulating pre-distributed public keys.
    \item \textbf{KEMTLS}: A pure post-quantum alternative \cite{kemtls} that substitutes handshake digital signatures with implicit KEM-based authentication.
\end{itemize}

Unlike prior evaluations relying on simplified linear calculations or single-run measurements, our validation framework integrates empirical micro-benchmarking, statistical error propagation, deterministic protocol-level constraint modeling, and stochastic network simulations. All cryptographic timings reported in this section were measured on the desktop-class x86 testbed described below and then propagated into protocol- and network-level models; we do not claim direct measurements on embedded IoT hardware in this work. The evaluation pipeline consists of four tightly coupled stages:
\begin{enumerate}
    \item \textbf{Empirical Micro-benchmarking}: Real cryptographic primitives are executed directly on the desktop-class physical testbed under an isolated environment to extract native performance distributions.
    \item \textbf{Robust Statistical Scaling}: For each cryptographic primitive, we execute $N = 10,000$ independent trials to compute the exact arithmetic mean, standard deviation, median, 95th percentile ($P95$), and 95\% confidence intervals.
    \item \textbf{Analytical Protocol Modeling}: Handshake execution is mapped using deterministic equations grounded in TLS 1.3 transport semantics \cite{rfc8446}, explicitly incorporating protocol-specific RTT multipliers and structural network trade-offs.
    \item \textbf{Monte Carlo Network Simulation}: The analytical framework is subjected to a stochastic simulation engine ($N = 10,000$ iterations) incorporating dynamic Exponential network jitter to evaluate stability, variance, and predictability under highly constrained topologies.
\end{enumerate}

\subsection{Experimental Setup}
\label{subsec:setup}

\subsubsection{Cryptographic Parameters}
All post-quantum primitives conform to the final NIST standardization track specifications: $ML\text{-}KEM\text{-}768$ (FIPS 203) for key encapsulation \cite{nist-fips203-Alagic} and $ML\text{-}DSA\text{-}65$ (FIPS 204) for digital signatures \cite{nistFIPS204}, both targeting NIST Security Level 3. The proposed HPQ-AKE protocol reuses legacy $RSA\text{-}3072$ public keys for implicit mutual authentication and can benefit from existing RSA acceleration where available, while embedding an ephemeral $ML\text{-}KEM\text{-}768$ layer to guarantee post-quantum forward secrecy. \rev{Handshake framing overhead is modeled using TLS 1.3 header conventions \mbox{\cite{rfc8446}}, rather than measured from a complete TLS implementation.}

\begin{table*}[pos=htbp]
\centering
\caption{Primitive Sizes and Public-Key Cryptographic Configurations}
\label{tab:params}
\begin{tabular}{lcc}
\toprule
\textbf{Primitive} & \textbf{Public Key (Bytes)} & \textbf{Signature / Ciphertext (Bytes)} \\
\midrule
ML-KEM-768 (Kyber) & 1,184 & 1,088 \\
ML-DSA-65 (Dilithium3) & 1,952 & 3,293 \\
RSA-3072 & 384 & 384 \\
X25519 & 32 & 32 \\
ECDSA-P256 & 64 & 64 \\
\bottomrule
\end{tabular}
\end{table*}

\noindent\textbf{\rev{Padding configuration.}}
\rev{The archived reference RSA benchmark uses RSA-3072 with PyCryptodome v3.20.0 defaults: OAEP-SHA1, MGF1-SHA1, and an empty label; each ciphertext is 384 bytes including OAEP encoding \mbox{\cite{rfc8017}}. OAEP randomness is independent of the NumPy simulation seed. The standalone benchmark uses a 28-byte test message, while the integrated demonstration uses 32 bytes. The network model takes aggregate payload sizes as inputs and adds MTU-dependent TCP/IP overhead, with no separate discretionary padding increment. AES-GCM execution and TLS record serialization are outside this aggregate model; nonce/tag lengths are therefore not additional measured simulation parameters. The explicit nonce policy in Section 3.3 specifies the revised protocol, not a configuration measured by the archived scripts; the historical payload and timing inputs are retained.}

\subsubsection{Gateway-to-Edge Backhaul Network Model}
The evaluation models the gateway-to-edge/cloud backhaul defined in the IoT deployment model, where RSA-enabled gateways terminate HPQ-AKE and communicate with edge or cloud services over TCP/TLS-compatible links. It does not model the local sensor-to-gateway radio segment; LPWAN, SCHC, IEEE 802.15.4, and device-level radio energy are outside the present scope. Accordingly, the backhaul model accounts for TLS Record Layer headers ($5\text{ Bytes}$ per record), Handshake message headers ($4\text{ Bytes}$ per message), public-key certificate chains, packet fragmentation bounds, and TCP window mechanics. The total handshake communication payload is defined as
\begin{equation}
S_{total} = \sum_{i=1}^{k} \text{size}(m_i) + \Omega_{TLS},
\label{eq:payload}
\end{equation}
where $m_i$ represents the raw public cryptographic keys or credential payloads, and $\Omega_{TLS}$ denotes the gapped transport framing, handshake headers, and certificate chain overheads.

Accounting for complete certificate validation chains under RFC 5280 \cite{rfc5280}, a standard $RSA\text{-}3072$ chain remains compact at approximately $1.4\text{ KB}$. In contrast, $ML\text{-}DSA\text{-}65$ chains expand beyond $5.4\text{ KB}$ due to large post-quantum signature allocations \cite{nistFIPS204}. Consequently, the final aggregate protocol payloads computed and fixed within our evaluation framework are:
\begin{itemize}
    \item \textbf{Hybrid TLS 1.3 (Full)}: 13,009 Bytes
    \item \textbf{Hybrid TLS 1.3 (Pre-cached)}: 6,209 Bytes
    \item \textbf{KEMTLS Baseline}: 12,172 Bytes
    \item \textbf{HPQ-AKE (Ours)}: 5,668 Bytes
\end{itemize}

The modeled backhaul implements a standard Ethernet Maximum Transmission Unit (MTU) restriction ($1,500\text{ Bytes}$) and references the TCP Initial Congestion Window (IW10) standard specified in RFC 6928 \cite{rfc6928}. Under IW10, the initial transmission budget is limited to 10 segments ($\approx 14.6\text{ KB}$). Because all evaluated protocol payloads fall within this limit, they can be dispatched within the initial congestion window flight without provoking structural multi-RTT penalty waves due to congestion window expansion.
\par

\begin{table*}[pos=htbp]
\centering
\caption{Gateway-to-Edge Backhaul Configuration Parameters}
\label{tab:network_params}
\begin{tabular}{lc}
\toprule
\textbf{Parameter} & \textbf{Value} \\
\midrule
Ethernet MTU & 1,500 Bytes \\
TCP Initial Congestion Window (IW10) & 10 Segments ($\approx 14.6\text{ KB}$) \\
Satellite Base RTT ($\text{RTT}_{base}$) & 600 ms \\
Bandwidth (Constrained Satellite Backhaul) & 50 kbps \\
Network Jitter ($\sigma_{jitter}$) & 12\% Exponential Noise \\
Monte Carlo Sampling Size ($N$) & 10,000 \\
\bottomrule
\end{tabular}
\end{table*}

\subsubsection{Hardware and Software Environment}
All local micro-benchmarks were systematically executed on an isolated hardware platform to eliminate external scheduling noise:
\begin{itemize}
    \item \textbf{CPU}: Intel Core i5-1135G7 @ 2.40 GHz (Turbo Boost disabled to maintain clock frequency stability)
    \item \textbf{RAM}: 8 GB DDR4 @ 3200 MHz
    \item \textbf{OS}: Ubuntu 22.04 LTS (Kernel 5.15)
    \item \texttt{Runtime:} Python 3.12.2
    \item {\RaggedRight\emergencystretch=3em \texttt{Core Cryptographic Libraries:} \texttt{liboqs v0.15.0}, \texttt{liboqs-python v0.14.1}, and \texttt{PyCryptodome v3.20.0}.\par}
\end{itemize}

\rev{The archived runners invoke RSA-OAEP through PyCryptodome without configuring dedicated RSA hardware offload. The x86 measurements therefore establish neither ARM-platform performance nor the benefit of a dedicated RSA accelerator.}
\rev{The standalone microbenchmark runner performs 100 warm-up iterations per routine followed by 10,000 timed iterations, using the high-resolution monotonic timer \mbox{\texttt{time.perf\_counter\_ns()}}. The separate integrated driver includes a 50-iteration demonstration benchmark; its network plots use the declared aggregate CPU profiles rather than automatically replacing them with that demonstration output.}
\par

\subsection{Statistically Rigorous Micro-benchmarking}
\label{subsec:microbenchmark}

\subsubsection{Measurement Methodology}
Each primitive was evaluated over $N = 10,000$ independent trials. To ensure empirical transparency, we avoid single-point estimations and report the full statistical spread: arithmetic mean ($\mu$), sample standard deviation ($\sigma$), median, 95th percentile ($P95$), and the 95\% confidence interval ($CI_{95}$). The confidence interval is calculated as
\begin{equation}
CI_{95} = 1.96 \cdot \frac{\sigma}{\sqrt{N}},
\end{equation}
where $N = 10,000$, the error margin of our sample mean is squeezed by a factor of $100$, ensuring exceptional statistical convergence. To compound individual primitive latencies into protocol-level CPU footprints, we apply the strict variance addition theorem for independent random variables (statistical error propagation)
\begin{equation}
\sigma_{total} = \sqrt{\sum_{i=1}^{n} \sigma_i^2}.
\label{eq:errorprop}
\end{equation}
This mathematically rigorous formulation ensures that localized hardware clock jitters do not introduce artificial linear accumulation flaws into our protocol-level evaluations.

\subsubsection{Primitive-Level Results}
Table \ref{tab:primitive_results} details the empirical benchmarks obtained directly from native execution on our hardware testbed.

\begin{table*}[pos=htbp]
\centering
\caption{Empirical Cryptographic Primitive Benchmark Results ($N = 10,000$)}
\label{tab:primitive_results}
\small
\begin{tabular}{lccccc}
\toprule
\textbf{Primitive} & \textbf{Mean (ms)} & \textbf{Std (ms)} & \textbf{Median (ms)} & \textbf{P95 (ms)} & \textbf{95\% CI} \\
\midrule
Kyber KeyGen & 0.0138 & 0.0209 & 0.0135 & 0.0170 & $\pm$0.00041 \\
Kyber Encap & \textbf{0.0157} & 0.0203 & 0.0150 & 0.0190 & $\pm$0.00039 \\
Kyber Decap & \textbf{0.0090} & 0.0014 & 0.0088 & 0.0110 & $\pm$0.00003 \\
\midrule
Dilithium Sign & 0.0884 & 0.0054 & 0.0872 & 0.0965 & $\pm$0.00011 \\
Dilithium Verify & 0.0276 & 0.0017 & 0.0271 & 0.0304 & $\pm$0.00003 \\
\midrule
RSA-3072 Encrypt & \textbf{0.6399} & 0.2100 & 0.6100 & 0.9500 & $\pm$0.00411 \\
RSA-3072 Decrypt & \textbf{2.7355} & 0.4709 & 2.7000 & 3.5000 & $\pm$0.00922 \\
\midrule
X25519 KeyGen & 0.0754 & 0.0042 & 0.0748 & 0.0812 & $\pm$0.00008 \\
X25519 ECDH & 0.0652 & 0.0031 & 0.0645 & 0.0701 & $\pm$0.00006 \\
ECDSA Sign & 0.6452 & 0.0384 & 0.6411 & 0.7025 & $\pm$0.00075 \\
ECDSA Verify & 1.3885 & 0.1914 & 1.3624 & 1.6841 & $\pm$0.00375 \\
\bottomrule
\end{tabular}
\end{table*}

The empirical data highlights a critical architectural trade-off: while lattice-based primitives ($ML\text{-}KEM$) feature optimal computational speeds ($\mu \le 0.0157\text{ ms}$) and tightly clustered variances, their deployment alongside $ML\text{-}DSA\text{-}65$ forces substantial transmission overheads due to signature scales. Conversely, $RSA\text{-}3072$ decryption exhibits an asymmetric computational penalty ($\mu = 2.7355\text{ ms}$), but its transmission footprint remains exceptionally compact.

\subsubsection{Protocol-Level CPU Latency}
The total software computational overhead ($T_{cpu}$) is formulated by aggregating the independent operational components required to execute a complete mutually authenticated handshake
\begin{equation}
T_{cpu} = \sum T_{KeyGen} + \sum T_{Encap/Decap} + \sum T_{Sign/Verify}.
\label{eq:cpu}
\end{equation}

\begin{table*}[pos=htbp]
\centering
\caption{Protocol-Level Combined CPU Execution Profiles}
\label{tab:cpu_total}
\resizebox{\textwidth}{!}{%
\begin{tabular}{llc}
\toprule
\textbf{Protocol} & \textbf{Cryptographic Operational Composition} & \textbf{Total CPU Time (ms) ($\mu \pm \sigma$)} \\
\midrule
Kyber-only Baseline & $\text{Kyber}_{KG} + \text{Kyber}_{Enc} + \text{Kyber}_{Dec} + \text{Sym}_{\text{agg}}$ & $0.0331 \pm 0.0057$ \\
TLS 1.3 (Classical) & $2\times\text{X25519}_{KG} + 2\times\text{ECDH} + \text{ECDSA}_{Sign} + \text{ECDSA}_{Verify} + \text{Sym}_{\text{agg}}$ & $2.1849 \pm 0.2426$ \\
Hybrid TLS 1.3 & $\text{Classical\_TLS}_{CPU} + \text{Kyber}_{KG} + \text{Kyber}_{Enc} + \text{Kyber}_{Dec} + \text{Dili}_{Sign} + \text{Dili}_{Verify} + \text{Sym}_{\text{agg}}$ & $2.3340 \pm 0.2497$ \\
KEMTLS & $2\times\text{Kyber}_{KG} + 2\times\text{Kyber}_{Enc} + 2\times\text{Kyber}_{Dec} + \text{Sym}_{\text{agg}}$ & $0.0661 \pm 0.0080$ \\
\textbf{HPQ-AKE (Ours)} & $\text{Kyber}_{KG} + 2\times\text{Kyber}_{Enc} + 2\times\text{Kyber}_{Dec} + 2\times\text{RSA}_{Enc} + 2\times\text{RSA}_{Dec} + \text{Sym}_{\text{agg}}$ & $\mathbf{7.1053 \pm 0.7036}$ \\
\bottomrule
\end{tabular}
}
\begin{flushleft}
\small \textbf{Note:} $\text{Sym}_{\text{agg}}$ represents the sub-millisecond symmetric cryptographic pipeline executed during the handshake, comprising SHA3 transcript hashing, HMAC-HKDF key derivation, and AEAD record protection. \rev{For reproducible network evaluation, the simulator takes the aggregate means and standard deviations shown here as explicit inputs; it does not separately execute or time the symmetric pipeline.}
\end{flushleft}
\end{table*}

As detailed in Table \ref{tab:cpu_total}, HPQ-AKE exhibits the highest overall local processing overhead ($7.1053 \pm 0.7036\text{ ms}$), driven by the two $RSA\text{-}3072$ private decapsulations required for mutual authentication. The gateway-class initiator defers its private operation until processing $M2$, so its first flight requires only Kyber encapsulation and RSA public-key encryption ($\approx 0.65\text{ ms}$ combined). This sequencing reduces first-flight cost but does not eliminate one private RSA operation at each endpoint; direct microcontroller performance is not measured in this work.

\subsection{Communication Efficiency Analysis}
The fundamental optimization vector of HPQ-AKE is to minimize communication overhead to safeguard network pipelines. By stripping out the massive post-quantum signature fields required by Hybrid TLS 1.3 and replacing them with a dual-KEM implicit authentication structure utilizing compact legacy RSA allocations, the protocol footprint shrinks significantly relative to the unoptimized Full handshake
\begin{equation}
\Delta_{reduction} = \frac{13,009\text{ B} - 5,668\text{ B}}{13,009\text{ B}} \times 100 \approx 56.4\%.
\end{equation}

\begin{figure*}[pos=htbp]
    \centering
    \includegraphics[width=0.92\linewidth]{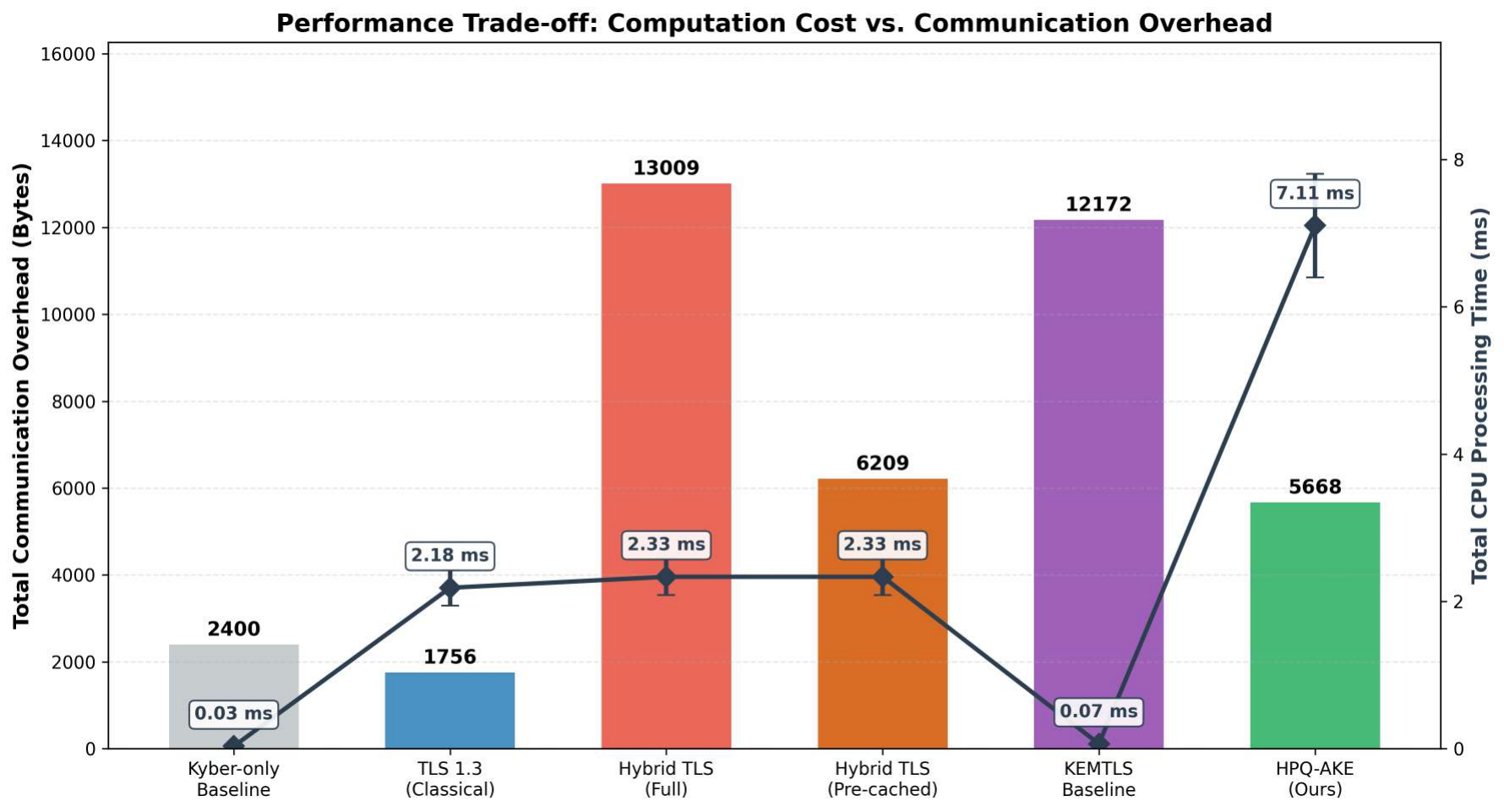}
    \caption{The architectural trade-off between communication overhead (payload size) and CPU execution time across post-quantum migration frameworks.}
    \label{fig:tradeoff}
\end{figure*}

As illustrated in the trade-off map (Figure \ref{fig:tradeoff}), HPQ-AKE is engineered around a deliberate architectural asymmetric exchange: it compromises a $4.77\text{ ms}$ increment in local processing time relative to Hybrid TLS 1.3 to eliminate exactly $7,341\text{ Bytes}$ of network transmission payload. Even when compared to the highly optimized Pre-cached Hybrid TLS baseline ($6,209\text{ Bytes}$), HPQ-AKE shaves off an additional MTU fragmentation frame in the modeled wire-level comparison. This architectural choice yields substantial modeled performance dividends in channels where bandwidth boundaries impose severe transmission latency penalties.

\begin{table*}[pos=htbp]
\centering
\caption{Handshake Payload Decomposition and Component Breakdown (in Bytes)}
\label{tab:payload_decomposition}
\resizebox{\textwidth}{!}{%
\begin{tabular}{lrrr}
\hline
\textbf{Handshake Component Layer} & \textbf{Hybrid TLS 1.3 (Full)} & \textbf{Hybrid TLS 1.3 (Pre-cached)} & \textbf{Proposed HPQ-AKE} \\ \hline
\textit{Core Cryptographic Primitives:} & & & \\
-- ML-KEM-768 Ephemeral Public Key ($pk_{eph}^{A}$)  & 1,184 & 1,184 & 1,184 \\
-- ML-KEM-768 Static Ciphertext ($C_{static}$)         & 1,088 & 1,088 & 1,088 \\
-- ML-KEM-768 Ephemeral Ciphertext ($C_{eph}$)       & 1,088 & 1,088 & 1,088 \\
-- RSA-3072 Asymmetric Ciphertext ($C_{RSA}$)          & 0     & 0     & 768   \\
-- Post-Quantum Certificates \& Signatures             & 8,400 & 1,600 & 0     \\
-- Auxiliary Data (Nonces, MACs, extensions)     & 461   & 461   & 356   \\ \hline
\textbf{Core Cryptographic Payload}                    & \textbf{12,221} & \textbf{5,421} & \textbf{4,484} \\ 
Protocol framing headers (Handshake / Record layer)   & 788   & 788   & 1,184 \\ \hline
\textbf{Total Application Payload ($S$)}                & \textbf{13,009} & \textbf{6,209} & \textbf{5,668} \\ \hline
Required MTU Network Fragments ($\lceil S / 1460 \rceil$) & 9 frames & 5 frames & 4 frames \\
Accumulated TCP/IP Wire Overhead ($40\text{ B} \times \text{frames}$) & $Z = 360$ & $X = 200$ & $Y = 160$ \\ \hline
\textbf{Total Wire-Level Footprint ($S^{wire}$)}        & \textbf{13,369} & \textbf{6,409} & \textbf{5,828} \\ \hline
\end{tabular}
}
\end{table*}
\noindent \textbf{A Note on Pre-cached Artifacts.} The 1,600 Bytes allocation for ``PQ Certs \& Sigs" under the Pre-cached baseline explicitly utilizes an optimized certificate profile based on the ML-DSA-44 (Dilithium2) static public key (1,312 Bytes per FIPS 204) and X.509 metadata (288 Bytes). While the Full Handshake baseline defaults to enterprise-grade ML-DSA-65 (NIST Level 3), the pre-cached asymmetric baseline is intentionally scoped to NIST Level 2. This architectural trade-off reflects pragmatic edge-computing environments where pre-cached assets are aggressively compressed to satisfy tight storage and memory constraints on remote hardware nodes.

While the proposed HPQ-AKE protocol significantly reduces the core cryptographic payload by eliminating heavy transcript signatures, a counterintuitive artifact manifests in Table~\ref{tab:payload_decomposition}, where the accumulated protocol framing headers total $1,184\text{ Bytes}$, exceeding the $788\text{ Bytes}$ required by the Hybrid TLS 1.3 baseline. This variance is strictly structural and is driven by two architectural requirements. First, unlike the streamlined 1-RTT structure of TLS 1.3, HPQ-AKE operates on a 1.5-RTT framework spanning three distinct network flights, which naturally incurs an additive accumulation of custom record encapsulation headers, message type descriptors, and length fields. Second, to achieve Mutual Implicit Authentication without explicit digital signatures, the protocol headers must explicitly transport extensive transcript-binding metadata, including cryptographic nonces ($32\text{ Bytes}$ each) and high-entropy session binders ($64\text{ Bytes}$), to enforce session context tying across all handshake flights. However, as demonstrated in the lower section of Table~\ref{tab:payload_decomposition}, this application-layer header expansion is completely offset at the network layer: by suppressing the massive signature fields, HPQ-AKE limits MTU fragmentation to only 4 frames, reducing the accumulated TCP/IP wire overhead to $Y = 160\text{ Bytes}$ and establishing the lowest absolute wire-level footprint ($5,828\text{ Bytes}$).

\subsection{Gateway-to-Edge Backhaul Latency Evaluation}
\label{subsec:network}

\subsubsection{Analytical Latency Model}
The end-to-end handshake execution latency ($T_{total}$) within a generic network topology is modeled through a physics-compliant multi-RTT equation
\begin{equation}
T_{total} = T_{cpu} + \frac{S_{total} \times 8}{BW} + (\text{RTT}_{base} \times \alpha_{rtt}),
\label{eq:latency_total}
\end{equation}
where $S_{total}$ represents the total payload size in Bytes, $BW$ is the channel bandwidth in kbps, $\text{RTT}_{base}$ is the network baseline round-trip propagation delay in milliseconds, and $\alpha_{rtt}$ is the structural protocol round-trip multiplier.

Our model addresses a critical physical constraint often omitted in simplified analyses: the handshake structural architecture dictates distinct round-trip counts under standard transport layer semantics:
\begin{itemize}
    \item \textbf{TLS 1.3 and Hybrid TLS 1.3} function natively as 1-RTT architectures ($\alpha_{rtt} = 1.0$).
    \item \textbf{KEMTLS} enforces an extended handshake dependency chain, requiring a 2.0-RTT sequence ($\alpha_{rtt} = 2.0$) \cite{kemtls}.
    \item \textbf{HPQ-AKE} sustains a 1.5-RTT sequence ($\alpha_{rtt} = 1.5$) to safely establish server authentication before completing mutual client validation.
\end{itemize}

\subsubsection{Latency Across Network Scenarios}
We evaluated the analytical performance across three modeled network profiles:
\begin{itemize}
    \item \textbf{Fiber/LAN Link}: $1\text{ Gbps}$, $\text{RTT} = 1\text{ ms}$ (representing high-speed edge environments).
    \item \textbf{4G/Wi-Fi Link}: $25\text{ Mbps}$, $\text{RTT} = 40\text{ ms}$ (representing standard mobile topologies).
    \item \textbf{Satellite-Assisted IoT Backhaul}: $50\text{ kbps}$, $\text{RTT} = 600\text{ ms}$ (representing narrow-band, high-propagation gateway-to-edge connectivity).
\end{itemize}

\begin{figure*}[pos=htbp]
    \centering
    \includegraphics[width=\linewidth]{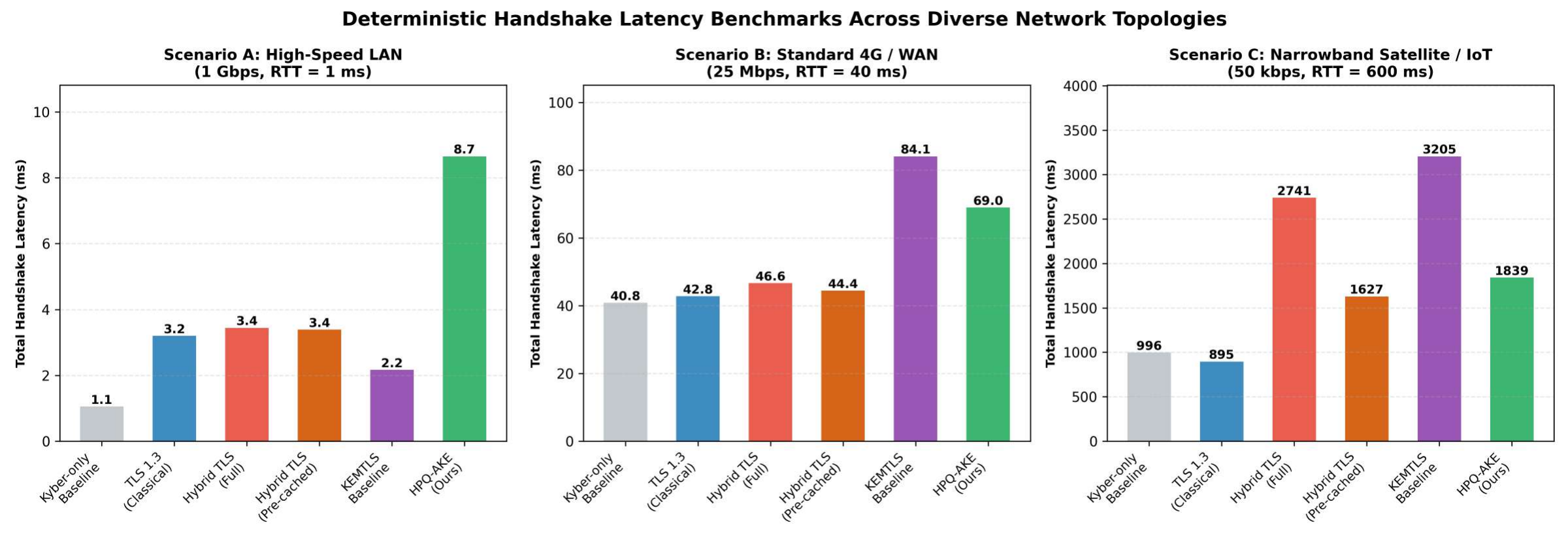}
    \caption{End-to-end handshake latency performance mapped across distinct network link profiles.}
    \label{fig:network_latency}
\end{figure*}

The experimental results shown in Figure \ref{fig:network_latency} reveal a sharp inversion in performance bottlenecks. In high-capacity Fiber/LAN pipelines, network transmission delay approaches zero, allowing local CPU execution times to dominate the performance profile; consequently, Hybrid TLS 1.3 maintains a lead.

As the modeled network transitions toward the narrow-band, high-latency Satellite/IoT-inspired profile, transmission delays grow substantially. While HPQ-AKE incurs a structural propagation lag over the $600\text{ ms}$ link due to its $1.5\text{-RTT}$ penalty, making it slightly slower than the Pre-cached TLS baseline, it bypasses the massive packet overhead limitations of the Hybrid TLS (Full) handshake and the heavy $2.0\text{-RTT}$ propagation penalty of KEMTLS.

\subsection{Monte Carlo Stability Analysis}
\rev{The reference routine in the official archive's \mbox{\texttt{HoanHaoTongHop.py}} initializes NumPy RandomState (MT19937) once with seed 42 and uses 10,000 samples per protocol at 50 kbps and 600 ms base RTT. In the script's fixed protocol order, normal CPU samples are drawn from the declared mean/std profiles in Table~\mbox{\ref{tab:cpu_total}} and clipped below at 0.001 ms, followed by exponential jitter samples. Jitter scale equals 0.12 times base RTT times the protocol's RTT multiplier; jitter affects propagation, not transmission time. This draw order reproduces the reported latency labels. The separate RTT-sensitivity experiment in Section~\mbox{\ref{subsec:break_even}} uses PCG64 with seeds 42--51.}
\par

\begin{figure*}[pos=htbp]
    \centering
    \includegraphics[width=0.9\linewidth]{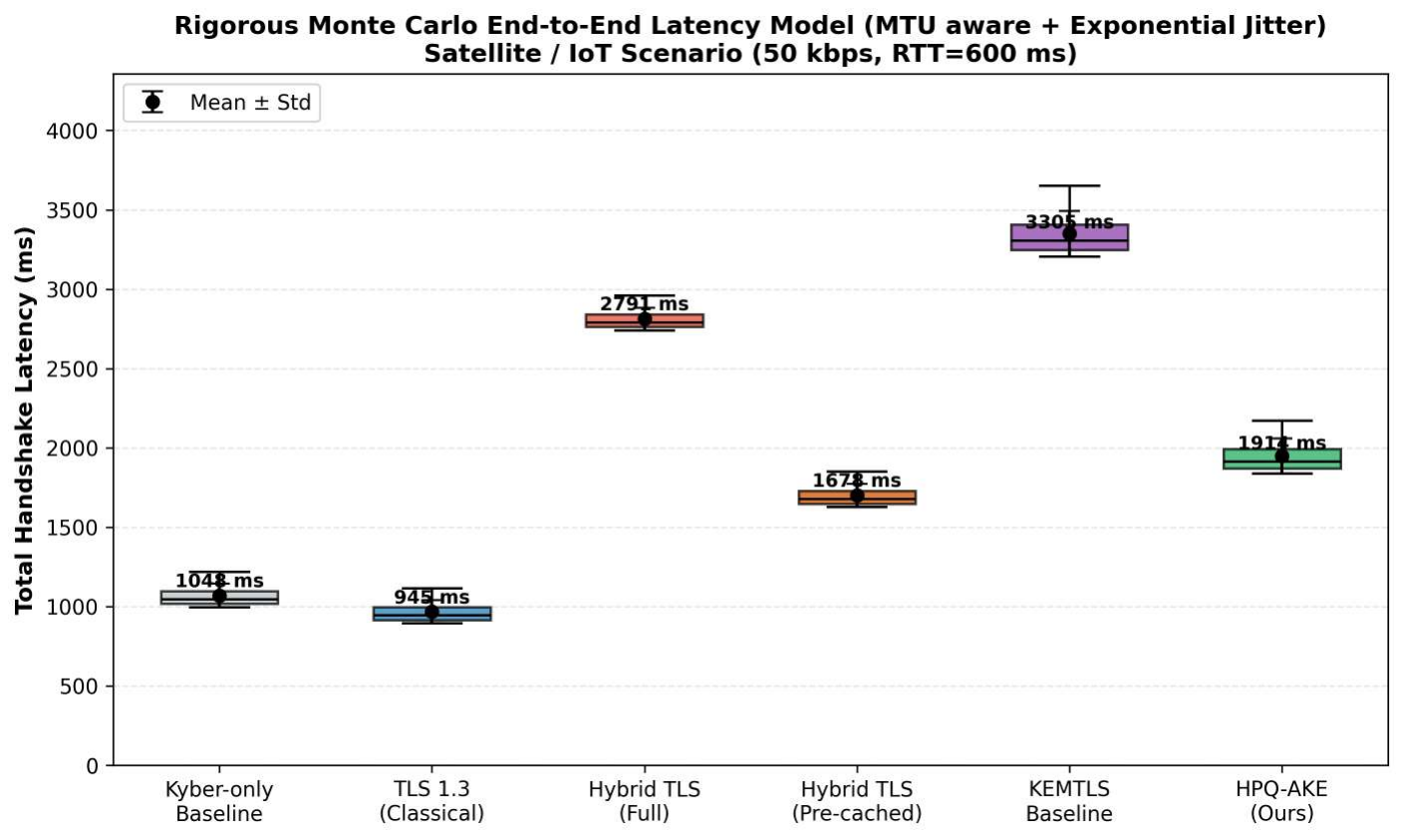}
    \caption{Handshake latency distributions under stochastic satellite network conditions ($50\text{ kbps}$, $\text{RTT} = 600\text{ ms}$, $12\%$ Exponential jitter, evaluated over $10,000$ runs).}
    \label{fig:boxplot}
\end{figure*}

Figure \ref{fig:boxplot} shows the resulting statistical distributions. \rev{Whiskers extend to the most extreme samples within 1.5 times the IQR of the respective quartiles. Outlier markers are hidden for readability, but samples are not removed from the numerical summaries.}

\rev{Replaying the reference routine yields full-precision medians of 2791.062850 ms (Full), 1678.138330 ms (Pre-cached), and 1914.732601 ms (HPQ-AKE). The following whole-millisecond values reproduce the figure labels, which use integer truncation:}
\begin{itemize}
    \item \textbf{Hybrid TLS 1.3 (Full) Median Latency}: $2,791\text{ ms}$
    \item \textbf{Hybrid TLS 1.3 (Pre-cached) Median Latency}: $1,678\text{ ms}$
    \item \textbf{HPQ-AKE (Ours) Median Latency}: $1,914\text{ ms}$
\end{itemize}

This demonstrates that under simulated stochastic satellite-like channels, HPQ-AKE achieves a stable modeled performance optimization against standard uncached deployments
\par
\begin{equation}
\begin{aligned}
\text{Latency Optimization \%}
&= \left(1 - \frac{1,914\text{ ms}}{2,791\text{ ms}}\right) \times 100 \\
&= 31.4\%.
\end{aligned}
\end{equation}

Although the $1.5\text{-RTT}$ structure introduces additional latency compared to the perfectly cached $1.0\text{-RTT}$ baseline, shrinking the transmission window to $5,668\text{ Bytes}$ effectively insulates the HPQ-AKE handshake process from severe packet arrival variance and stochastic channel degradation, delivering a tight and predictable latency dispersion profile.

\subsection{Break-Even Bandwidth Analysis}
\label{subsec:break_even}
\rev{Communication efficiency is also a central concern in asynchronous federated learning for synthetic aperture radar (SAR) image target recognition, as illustrated by AsyncFed-SAR~\mbox{\cite{AsyncFedSAR}}. In the present work, we examine a different, protocol-level question: how RTT and jitter affect HPQ-AKE's break-even bandwidth. Our sensitivity analysis is derived from the HPQ-AKE latency model, rather than adopting a federated-learning performance model.}

To establish the exact boundary at which our communication payload reduction overcomes the local $RSA\text{-}3072$ processing overhead and structural round-trip delays, we solve for the crossover point where the total handshake latency of Hybrid TLS 1.3 equals that of HPQ-AKE. By incorporating our multi-RTT physics model, the equilibrium equation is formulated using wire-level payloads ($S^{wire}$) to account for network-layer overheads
\begin{equation}
\begin{aligned}
\frac{8\left(S^{wire}_{TLS} - S^{wire}_{HPQ}\right)}{BW}
&= T_{cpu}^{HPQ} - T_{cpu}^{TLS} \\
&\quad + 0.5\,\text{RTT}_{crossover}.
\end{aligned}
\label{eq:breakeven_theory}
\end{equation}

Crucially, as established in Section \ref{subsec:setup}, since all protocol payloads are fully contained within the TCP IW10 envelope, they execute without triggering congestion-driven transport delays, maintaining standard round-trip structures. Isolating the operational channel crossover bandwidth ($BW_{break-even}$) yields

\begin{equation}
BW_{break-even} = \frac{8 \cdot (S^{wire}_{TLS} - S^{wire}_{HPQ})}{\Delta T_{penalty}},
\label{eq:breakeven_final}
\end{equation}

where the structural penalty gap $\Delta T_{penalty}$ is defined as
\begin{equation*}
\Delta T_{penalty} = (T_{cpu}^{HPQ} - T_{cpu}^{TLS}) + (\Delta\alpha_{rtt} \cdot \text{RTT}_{crossover})
\end{equation*}
and $\Delta\alpha_{rtt} = \alpha_{rtt}^{HPQ} - \alpha_{rtt}^{TLS}$ represents the round-trip structural difference.

Evaluating Equation \ref{eq:breakeven_final} using our empirical hardware means ($T_{cpu}^{HPQ} = 7.1053\text{ ms}$, $T_{cpu}^{TLS} = 2.3340\text{ ms}$) under standard evaluation IoT network configurations ($\text{RTT}_{crossover} = 20\text{ ms}$, $\Delta\alpha_{rtt} = 0.5$) calculates a penalty gap of $\Delta T_{penalty} \approx 14.77\text{ ms}$. To provide a comprehensive evaluation, we solve Equation \ref{eq:breakeven_final} against two distinct baselines manifested in network deployments, where wire payloads are explicitly derived from core payloads ($S$) and TCP/IP header overheads ($\text{MTU} = 1500\text{ B}$, $\text{MSS} = 1460\text{ B}$) across fragmented frames:
\par

\begin{itemize}
    \item \textbf{Vs. Hybrid TLS 1.3 (Pre-cached Baseline)}: With an effective wire-level payload difference of $581\text{ Bytes}$ ($4,648\text{ bits}$), where the wire sizes are explicitly adjusted for MTU overhead ($S^{wire}_{TLS} = 6,409\text{ Bytes}$ vs. $S^{wire}_{HPQ} = 5,828\text{ Bytes}$). The operational threshold is identified at
    \begin{equation}
    BW_{break-even}^{\text{pre-cached}} \approx 314.7\text{ kbps} \approx 0.31\text{ Mbps}.
    \end{equation}
    
    \item \textbf{Vs. Hybrid TLS 1.3 (Full Handshake Baseline)}: When certificates and post-quantum signatures are transmitted in full, the wire-level footprint expands due to higher fragmentation overhead ($S^{wire}_{TLS} = 13,009 + 360 = 13,369\text{ Bytes}$ vs. $S^{wire}_{HPQ} = 5,668 + 160 = 5,828\text{ Bytes}$), yielding a wire payload gap of $7,541\text{ Bytes}$ ($60,328\text{ bits}$). This shifts the operational break-even boundary up to\par
    \begin{equation}
    BW_{break-even}^{\text{full}} \approx \text{\rev{4,084.1}}\text{ kbps} \approx 4.08\text{ Mbps}.
    \end{equation}
\end{itemize}

\begin{figure*}[pos=htbp]
    \centering
    \includegraphics[width=\textwidth]{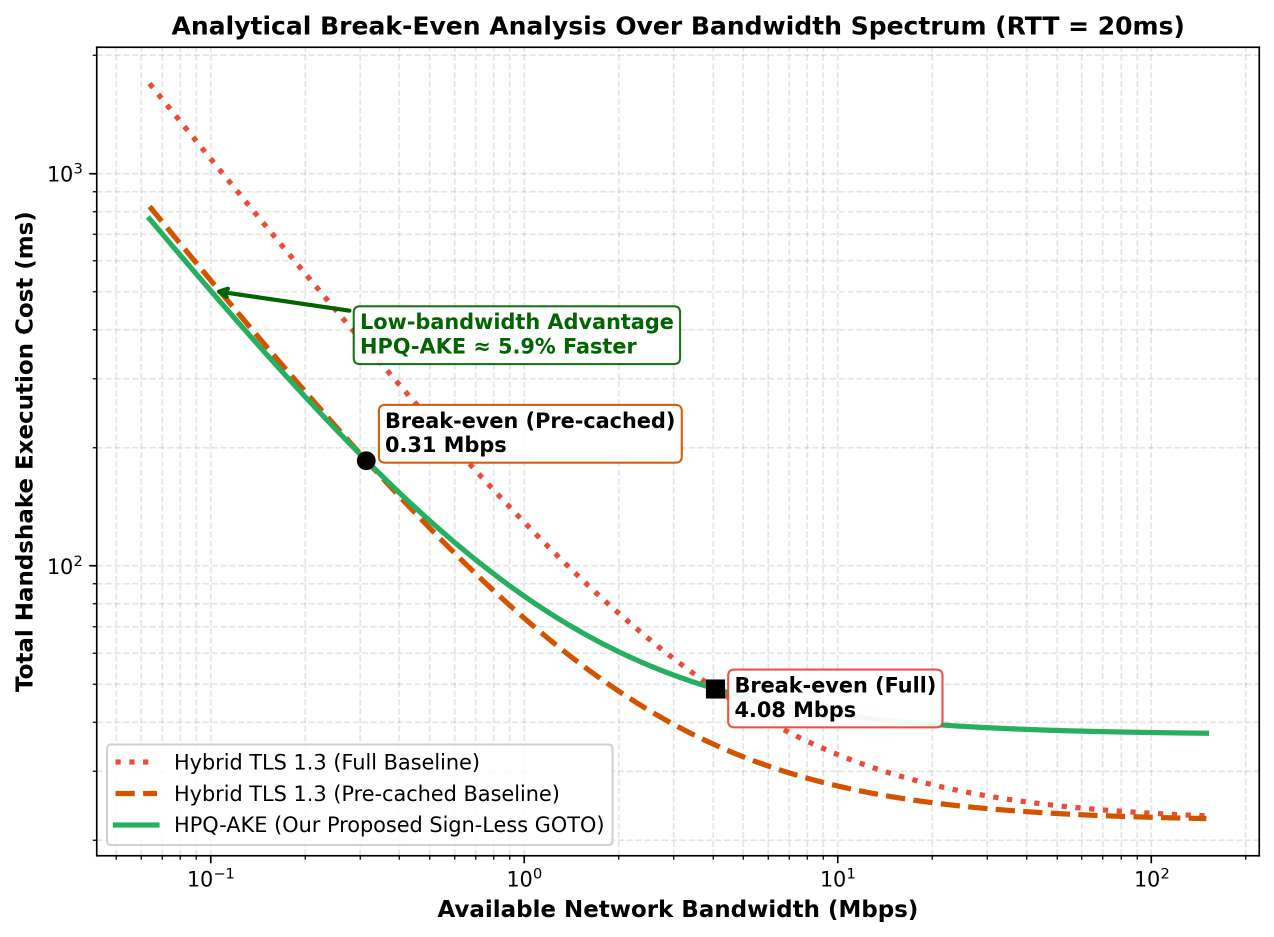}
    \caption{Break-even bandwidth operational boundary analysis mapping HPQ-AKE optimization zones against both Full and Pre-cached Hybrid TLS 1.3 baselines.}
    \label{fig:break_even}
\end{figure*}

\rev{The 0.31 Mbps and 4.08 Mbps thresholds are conditional on the 20 ms RTT setting, not universal bandwidth limits. Holding wire sizes and CPU means fixed in Equation~\mbox{\ref{eq:breakeven_final}}, increasing RTT increases the additional half-RTT cost of HPQ-AKE and monotonically lowers both crossover bandwidths. Table~\mbox{\ref{tab:rtt_sensitivity}} quantifies this dependence over 1--1,000 ms. At 600 ms without jitter, the thresholds are 0.0153 Mbps against Pre-cached and 0.1979 Mbps against Full; thus a 50 kbps link lies in the advantage region only against Full.}

\noindent\textbf{\rev{Seeded RTT and jitter sensitivity.}}
\rev{We newly evaluate 10,000 samples per seed using NumPy Generator with PCG64 and integer seeds 42--51. For each seed, one array of unit-scale exponential variates is reused across both baselines, RTT settings, and jitter fractions, providing paired comparisons. The realized RTT equals the base RTT multiplied by one plus the jitter fraction times the exponential variate. Fractions 0, 0.12, and 0.30 represent no jitter, the nominal stress level, and a higher-jitter stress test; the added delay is positive, with mean and standard deviation equal to the fraction times base RTT. CPU means remain fixed at 7.1053 ms for HPQ-AKE and 2.3340 ms for both TLS baselines to isolate RTT effects; no empirical CPU resampling is claimed.}

\begin{table*}[pos=htbp]
\centering
\caption{\rev{New RTT-sensitivity analysis: crossover bandwidths in Mbps. Stochastic entries are pooled medians from ten seeds with 10,000 samples each; zero-jitter entries are deterministic.}}
\label{tab:rtt_sensitivity}
\small
\begin{tabular}{rrrrrrr}
\toprule
\textbf{\rev{Base RTT}} & \multicolumn{2}{c}{\textbf{\rev{No jitter}}} & \multicolumn{2}{c}{\textbf{\rev{12\% jitter}}} & \multicolumn{2}{c}{\textbf{\rev{30\% jitter}}} \\
\textbf{\rev{(ms)}} & \rev{Pre-cached} & \rev{Full} & \rev{Pre-cached} & \rev{Full} & \rev{Pre-cached} & \rev{Full} \\
\midrule
\rev{1} & \rev{0.8818} & \rev{11.4446} & \rev{0.8748} & \rev{11.3549} & \rev{0.8647} & \rev{11.2230} \\
\rev{20} & \rev{0.3147} & \rev{4.0841} & \rev{0.2979} & \rev{3.8662} & \rev{0.2758} & \rev{3.5796} \\
\rev{50} & \rev{0.1561} & \rev{2.0264} & \rev{0.1459} & \rev{1.8939} & \rev{0.1329} & \rev{1.7249} \\
\rev{100} & \rev{0.0849} & \rev{1.1015} & \rev{0.0789} & \rev{1.0236} & \rev{0.0713} & \rev{0.9256} \\
\rev{200} & \rev{0.0444} & \rev{0.5758} & \rev{0.0411} & \rev{0.5334} & \rev{0.0370} & \rev{0.4804} \\
\rev{400} & \rev{0.0227} & \rev{0.2946} & \rev{0.0210} & \rev{0.2725} & \rev{0.0189} & \rev{0.2448} \\
\rev{600} & \rev{0.0153} & \rev{0.1979} & \rev{0.0141} & \rev{0.1829} & \rev{0.0127} & \rev{0.1643} \\
\rev{1000} & \rev{0.0092} & \rev{0.1195} & \rev{0.0085} & \rev{0.1104} & \rev{0.0076} & \rev{0.0991} \\
\bottomrule
\end{tabular}
\end{table*}

\rev{At 600 ms, median crossover bandwidths fall to 0.0141/0.1829 Mbps (Pre-cached/Full) with 12\% jitter and 0.0127/0.1643 Mbps with 30\% jitter. Across the evaluated settings, the range of seed-specific medians is below 1.2\% of the pooled median. The supplement provides the script \mbox{\texttt{analysis/rtt\_sensitivity.py}}, runtime versions, input manifest, per-seed results, and pooled 5th/95th percentiles. Quantiles use linear interpolation without outlier trimming; percentile ranges describe modeled variability, not confidence intervals. This experiment tests sensitivity of the stated crossover model, not exact reproduction of Figure~\mbox{\ref{fig:boxplot}} or field measurements.}
\par

\subsection{DoS Mitigation and CPU Saturation Bounds}
While HPQ-AKE dramatically reduces transmission latency, its Dual-Layered KEM architecture also provides a critical defense against cryptographic denial-of-service (DoS) attacks. In a standard PKI handshake, an adversary can flood the responder with garbage initializations, forcing the server to execute heavy asymmetric operations (e.g., RSA decryption or signature verification) and rapidly exhausting CPU resources.

HPQ-AKE mitigates this through its Layer-1 filter. Before the responder attempts to process the heavy $RSA\text{-}3072$ decryption ($\mu = 2.7355\text{ ms}$), it must first successfully decapsulate the static Kyber ciphertext ($C_{static}$) and verify $MAC_1$. This initial post-quantum decapsulation is highly asymmetric in terms of computational cost ($\mu = 0.0090\text{ ms}$). 

We can formally define the CPU saturation bounds per processing core. Without the early mitigation filter, a direct RSA-based exhaustion attack would saturate the core at
\begin{equation}
\text{Capacity}_{RSA} = \frac{1}{2.7355 \text{ ms}} \approx 365 \text{ handshakes/second}.
\end{equation}
By forcing the attacker to pass the Kyber decapsulation filter, the responder's drop-rate capacity scales up to
\begin{equation}
\text{Capacity}_{Filter} = \frac{1}{0.0090 \text{ ms}} \approx 111,111 \text{ packets/second}.
\end{equation}
This architecture delivers a \textbf{$\approx 304\times$} reduction in computation cost. It is important to note that this mechanism is not a perfectly stateless DoS shield (such as TLS 1.3 Cookies), as Kyber decapsulation remains a cryptographic operation. A volumetric flood exceeding $111,000\text{ pps}$ per core will still induce CPU saturation. However, by amortizing the asymmetric cost, HPQ-AKE forces the adversary to expend vastly more network bandwidth to achieve computational exhaustion, thereby shifting the bottleneck from the local CPU back to the network perimeter.
\subsection{Discussion}
The combined micro-benchmarking and simulation results support our core architectural hypothesis: within bandwidth-constrained modeled networks, reducing transmission payload scale can yield greater latency reductions than minimizing local CPU execution times alone. While the introduction of $RSA\text{-}3072$ increases server-side computation by a few milliseconds on the measured x86 testbed, this cost is offset in the model by the elimination of large post-quantum signature transmissions, which prevents fragmentation overhead and packet delivery variance.
\par

\rev{RSA imposes a transition boundary rather than permanent authentication security. Section~\mbox{\ref{sec:hybrid_secrecy}} separates secrecy of previously established fresh sessions, subject to the stated erasure assumptions, from active authentication that still relies on classical RSA hardness. ML-KEM therefore does not justify claiming mutual authentication against an adversary able to break RSA during the handshake. Deployment planning should include retirement of RSA authentication and any RSA-based provisioning trust anchors; migration to post-quantum replacements requires renewed protocol analysis, not merely a larger RSA modulus.}

\rev{Legacy RSA support is also conditional: deployments must verify the required OAEP parameters and permitted key usage rather than assume that signing-only credentials can be reused for decapsulation \mbox{\cite{rfc5280KeyUsage}}. Protected private-key operations and OAEP failure handling without distinguishable error or timing leakage remain implementation requirements \mbox{\cite{rfc8017}}. Thus, compatible endpoint software and credential provisioning are prerequisites, not capabilities established by the x86 benchmark or a drop-in TLS replacement.}

\noindent\textbf{\rev{Side-channel implementation risks.}}
\rev{The protocol analysis does not model timing or power leakage from cryptographic operations; algorithmic security must be distinguished from implementation resistance to such attacks \mbox{\cite{Zhou}}. For RSA, deployment should combine blinding with constant-time private-key arithmetic and the OAEP error-handling discipline noted above. For ML-KEM (Kyber), decapsulation requires scrutiny of secret-dependent branches, memory accesses, arithmetic timing, and ciphertext comparison/implicit-rejection selection in the compiled target code. Constant-time execution alone does not establish resistance to power analysis; masking sensitive intermediates and other platform-specific protections require separate evaluation. Repeated use of static keys makes accumulated leakage particularly relevant, while ephemeral-key erasure does not prevent leakage during execution. These are deployment requirements and evaluation targets, not protections demonstrated by the present x86 benchmarks or network simulations.}

\rev{Operational suitability depends on more than handshake bytes. The gateway is a trusted termination point, not protection against gateway compromise or insecurity of the local sensor link. Pre-cached-key maintenance adds the availability and management costs discussed in Section 4.2, outside the reported handshake accounting. Higher RTT and effective certificate caching can favor the Pre-cached baseline. Target-device latency, energy, memory, concurrent-handshake capacity, and physical side-channel resistance still require deployment-specific validation. For an ARM target, both HPQ-AKE and baseline CPU profiles must be remeasured and the modeled latency and break-even thresholds recomputed. Gateway-class hardware or RSA acceleration alone does not establish the reported speedups. The wire-byte accounting remains a property of the specified message formats, not evidence of platform-independent execution performance.}

\section*{Reproducibility}
\rev{The experimental source code is publicly available in the official GitHub repository.}\footnote{\url{https://github.com/HPQ-Protocols/HPQ-AKE-Protocol}}
\rev{The official source snapshot is \mbox{\texttt{HPQ-AKE-Protocol-main.zip}}, with \mbox{\texttt{HoanHaoTongHop.py}} selected for the current network results. The supplement provides NumPy-only replay and RTT-sensitivity scripts, archive/source checksums, runtime versions, and machine-readable outputs; execution instructions and version selection are documented in \mbox{\texttt{analysis/README.md}}. Network replay is distinct from rerunning the original hardware microbenchmarks.}

\section{Conclusion and Future Work}
\label{sec:conclusion}
\subsection{Conclusion}
This paper proposed \textbf{HPQ-AKE}, a hybrid authenticated key exchange protocol designed for the efficient migration from classical Public Key Infrastructure (PKI) to post-quantum key establishment. By eliminating heavy transcript signatures within the handshake phase, the protocol leverages a dual-layer Key Encapsulation Mechanism (KEM) construction. In this architecture, ML-KEM provides post-quantum session-key confidentiality, while RSA-OAEP facilitates Mutual Implicit Authentication by reusing existing, widely deployed legacy RSA public keys.

Desktop-class software implementation and empirical micro-benchmarking using the \texttt{liboqs} and \texttt{PyCryptodome} ecosystems indicate that securing mutual authentication via dual RSA decryption operations requires a total computational footprint of 7.11 ms under the measured x86 evaluation environment. In exchange for this computational trade-off, HPQ-AKE requires $4,484\text{ Bytes}$ of core cryptographic payload. As illustrated by Table~\ref{tab:payload_decomposition}, protocol framing increases the total handshake transmission to $5,668\text{ Bytes}$, while TCP/IP headers produce a final wire-level footprint of $5,828\text{ Bytes}$. Relative to the $13,009\text{ Bytes}$ Hybrid TLS 1.3 Full handshake payload, HPQ-AKE therefore achieves a 56.4\% protocol-level reduction and eliminates an entire MTU fragmentation frame in the modeled comparison, even against the Pre-cached Hybrid TLS baseline ($6,209\text{ Bytes}$).
\par

Stochastic network-level Monte Carlo simulations over $10,000$ iterations characterize the physical bounds of this architectural trade-off. Under a modeled narrow-band $50\text{ kbps}$ satellite-like configuration with a $600\text{ ms}$ propagation delay, HPQ-AKE's necessary $1.5\text{-RTT}$ structural requirement yields a simulated median latency of 1,914 ms. While this propagation penalty makes it slower than a perfectly Pre-cached TLS baseline in high-latency scenarios, HPQ-AKE still delivers a 31.4\% modeled optimization against the standard Hybrid TLS 1.3 Full handshake (2,791 ms). \rev{At 20 ms RTT, deterministic break-even bandwidths are 0.31 Mbps against Pre-cached and 4.08 Mbps against Full; the added RTT-sensitivity analysis shows that these fall to 0.0153 and 0.1979 Mbps, respectively, at 600 ms without jitter. These conditional thresholds reflect the balance between wire-byte savings, computation, and the additional half-RTT requirement.} At $100\text{ kbps}$, HPQ-AKE executes 5.9\% faster than the strongest pre-cached baseline in our model. \rev{These findings identify RSA-KEM-based implicit authentication as a candidate for gateway-class evaluation under the measured x86 CPU assumptions, not as validated performance on ARM gateways or low-end IoT endpoints.}

\subsection{Future Work}
\rev{Further platform-specific validation is needed before deployment; the following research and engineering directions remain open:}
\begin{itemize}
    \item \textbf{Computer-Verified Formal Verification}: Although this paper establishes a rigorous, game-based formal reduction in the Quantum Random Oracle Model (QROM), manual cryptographic proofs remain susceptible to subtle human errors. Future work will focus on modeling the HPQ-AKE execution flow inside automated verification frameworks, such as EasyCrypt or CryptoVerif, to derive machine-checked mathematical guarantees of its post-quantum security properties.
    \item \textbf{\rev{ARM Gateway and MCU Validation}}: \rev{Future work will benchmark the revised HPQ-AKE and the same baselines on ARM Cortex-A gateway platforms and STM32F4-class Cortex-M4 boards, comparing software RSA with supported hardware offload where available. Measurements of latency, cycles, RAM/flash, and energy will determine whether each target can execute the full handshake or requires acceleration or gateway termination; no current ARM performance result is claimed.}
    \item \textbf{Standardization and Side-Channel Resilience}: We aim to explore the formal integration paths of the HPQ-AKE handshake model into existing protocol standards, such as TLS 1.3 and SSH, by formalizing explicit ciphersuite specifications. \rev{Future target-board evaluation will combine statistical timing-leakage tests with power-trace analysis of RSA and ML-KEM decapsulation in the integrated handshake, including malformed-ciphertext processing. We will assess blinding and masking countermeasures and their latency, memory, and energy costs. These measurements are planned work, not established side-channel guarantees.}
    \item  \textbf{Evolution Toward a Pure Post-Quantum KEM-Based Framework}: As legacy asymmetric cryptographic acceleration engines are progressively phased out across modern network boundaries, we plan to evolve the current hybrid architecture into a pure post-quantum, sign-less authentication framework. This evolutionary step involves substituting the legacy RSA authentication anchor with static post-quantum KEM key pairs (e.g., using long-term ML-KEM keys), thereby ensuring comprehensive, long-term cryptographic hardening against fault-tolerant quantum adversaries well beyond the migration era.\par
\end{itemize}

\ifdefined\BlindManuscript
\else
\section*{Acknowledgment}
The authors would like to express their sincere gratitude to the anonymous reviewers for their thorough evaluations, detailed comments, and constructive suggestions, which have significantly improved the technical rigor, clarity, and overall quality of this work.

Khiem Pham-Tuan wishes to thank Ho Chi Minh City Open University and Ho Chi Minh City University of Industry and Trade (HUIT) for facilitating this research and providing the computing resources necessary for the experiments. Similarly, Minh Quang Le thanks Ho Chi Minh City Open University for supporting and facilitating this study. Khuong Nguyen-An acknowledges Ho Chi Minh City University of Technology (HCMUT), VNU-HCM, for supporting his study.

\printcredits
\fi

\section*{Declaration of Generative AI usage:} During the preparation of this manuscript, the authors utilized generative AI tools solely to improve the grammatical structure and readability of the text, taking full responsibility for the final academic integrity.

\bibliographystyle{cas-model2-names}

\renewcommand{\bibfont}{\fontsize{7.8pt}{9.5pt}\selectfont}
\setlength{\bibsep}{2pt}
\begingroup
\newif\ifRevisedReference
\newcommand{\EndRevisedReference}{%
  \ifRevisedReference\end{mdframed}\RevisedReferencefalse\fi}
\let\OriginalBibitem\bibitem
\renewcommand{\bibitem}[2][]{%
  \EndRevisedReference
  \OriginalBibitem[#1]{#2}%
  \ifboolexpr{
    test {\ifstrequal{#2}{AWSGreengrassFAQ}} or
    test {\ifstrequal{#2}{AWSGreengrassAuthentication}} or
    test {\ifstrequal{#2}{InmarsatErgonEnergy}} or
    test {\ifstrequal{#2}{rfc5280KeyUsage}} or
    test {\ifstrequal{#2}{Ducas}} or
    test {\ifstrequal{#2}{Raj}} or
    test {\ifstrequal{#2}{Schwabe2021}} or
    test {\ifstrequal{#2}{Zhang}} or
    test {\ifstrequal{#2}{AsyncFedSAR}}
  }{%
    \RevisedReferencetrue
    \begin{mdframed}[backgroundcolor=white,hidealllines=true,
      skipabove=0pt,skipbelow=0pt,innertopmargin=0pt,
      innerbottommargin=0pt,innerleftmargin=0pt,innerrightmargin=0pt]%
  }{}%
}
\pretocmd{\endthebibliography}{\EndRevisedReference}{}{}
\bibliography{cas-refs}

@INPROCEEDINGS{Shor,
  author    = {P. W. Shor},
  title     = {Algorithms for quantum computation: discrete logarithms and factoring},
  booktitle = {Proc. 35th Annu. Symp. Found. Comput. Sci.},
  year      = {1994},
  pages     = {124--134}
}

@TECHREPORT{nist-report2022,
  author      = {G. Alagic and others},
  title       = {Status report on the third round of the NIST post-quantum cryptography standardization process},
  institution = {NIST},
  year        = {2022}
}

@ARTICLE{Avanzi,
  author  = {R. Avanzi and others},
  title   = {CRYSTALS-Kyber algorithm specifications and supporting documentation},
  journal = {NIST PQC Round},
  volume  = {2},
  number  = {4},
  year    = {2019},
  pages   = {1--43}
}

@INPROCEEDINGS{Ducas,
  author    = {L. Ducas and T. Lepoint and V. Lyubashevsky and P. Schwabe and G. Seiler and D. Stehlé},
  title     = {Crystals--dilithium: Digital signatures from module lattices},
  booktitle = {NIST PQC Round},
  year      = {2018}
}

@MISC{nist-fips203-Alagic,
  author       = {G. Alagic and others},
  title        = {Module-lattice-based key-encapsulation mechanism standard},
  howpublished = {FIPS PUB 203},
  month        = {Aug.},
  year         = {2024}
}

@MISC{Stebila,
  author       = {D. Stebila and S. Fluhrer and S. Gueron},
  title        = {Hybrid key exchange in TLS 1.3},
  howpublished = {IETF, Draft},
  year         = {2020}
}

@INPROCEEDINGS{7-Bos,
  author    = {J. W. Bos and others},
  title     = {Post-quantum key exchange for the TLS protocol from the ring learning with errors problem},
  booktitle = {Proc. IEEE Symp. Secur. Privacy (SP)},
  year      = {2015},
  pages     = {572--589}
}

@INPROCEEDINGS{kemtls,
  author    = {P. Schwabe and D. Stebila and T. Wiggers},
  title     = {Post-quantum TLS without handshake signatures},
  booktitle = {Proc. ACM SIGSAC Conf. Comput. Commun. Secur.},
  year      = {2020},
  pages     = {1461--1480}
}

@MISC{pqm4-Kannwischer,
  author = {M. J. Kannwischer and J. Rijneveld and P. Schwabe and K. Stoffelen},
  title  = {pqm4: Testing and benchmarking NIST PQC on ARM Cortex-M4},
  year   = {2019}
}

@ARTICLE{9-Battarbee,
  author  = {C. Battarbee and C. Striecks and L. Perret and S. Ramacher and K. Verhaeghe},
  title   = {Quantum-safe hybrid key exchanges with KEM-based authentication},
  journal = {EPJ Quantum Technol.},
  volume  = {12},
  number  = {1},
  pages   = {128},
  year    = {2025}
}

@ARTICLE{Astrizi,
  author  = {T. L. Astrizi and R. Custódio},
  title   = {Seamless transition to post-quantum TLS 1.3: A hybrid approach using identity-based encryption},
  journal = {Sensors},
  volume  = {24},
  number  = {22},
  pages   = {7300},
  year    = {2024}
}

@ARTICLE{V.Shoup,
  author  = {V. Shoup},
  title   = {A proposal for an ISO standard for public key encryption},
  journal = {Cryptology ePrint Archive},
  year    = {2001}
}

@INPROCEEDINGS{BosEtAl,
  author    = {J. Bos and others},
  title     = {CRYSTALS-Kyber: a CCA-secure module-lattice-based KEM},
  booktitle = {Proc. IEEE Eur. Symp. Secur. Privacy (EuroS\&P)},
  year      = {2018},
  pages     = {353--367}
}

@INPROCEEDINGS{b1,
  author    = {M. Bellare and P. Rogaway},
  title     = {Entity authentication and key distribution},
  booktitle = {Proc. Annu. Int. Cryptol. Conf.},
  year      = {1993},
  pages     = {232--249}
}

@MISC{nistFourthRound,
  author = {G. Alagic},
  title  = {Status report on the fourth round of the NIST post-quantum cryptography standardization process},
  year   = {2025}
}

@ARTICLE{Alagic2025,
  author  = {G. Alagic and F. Bajaj and A. Kocoglu},
  title   = {The best of both KEMs: Securely combining KEMs in post-quantum hybrid schemes},
  journal = {Cryptology ePrint Archive},
  year    = {2025}
}

@ARTICLE{Ricci,
  author  = {S. Ricci and others},
  title   = {Hybrid keys in practice: Combining classical, quantum and post-quantum cryptography},
  journal = {IEEE Access},
  volume  = {12},
  pages   = {23206--23219},
  year    = {2024}
}

@INPROCEEDINGS{Raj,
  author    = {A. Raj and V. Balachandran},
  title     = {A hybrid encryption framework combining classical, post-quantum, and QKD methods},
  booktitle = {Proc. Int. Conf. Appl. Cryptogr. Netw. Secur.},
  year      = {2025}
}

@ARTICLE{E.Crockett,
  author  = {E. Crockett and C. Paquin and D. Stebila},
  title   = {Prototyping post-quantum and hybrid key exchange and authentication in TLS and SSH},
  journal = {Cryptology ePrint Archive},
  year    = {2019}
}

@INPROCEEDINGS{Zhang,
  author    = {J. Zhang and others},
  title     = {Authenticated key exchange from ideal lattices},
  booktitle = {Proc. Annu. Int. Conf. Theory Appl. Cryptographic Techn.},
  year      = {2015}
}

@ARTICLE{Xia,
  author  = {T. Xia and others},
  title   = {A quantum-resistant identity authentication and key agreement scheme for UAV networks based on Kyber algorithm},
  journal = {Drones},
  volume  = {8},
  number  = {8},
  pages   = {359},
  year    = {2024}
}

@INPROCEEDINGS{Schwabe2021,
  author    = {P. Schwabe and D. Stebila and T. Wiggers},
  title     = {More efficient post-quantum KEMTLS with pre-distributed public keys},
  booktitle = {Proc. Eur. Symp. Res. Comput. Secur.},
  year      = {2021}
}

@MISC{rfc5990,
  author       = {J. Randall},
  title        = {Use of the RSA-KEM key transport algorithm in the cryptographic message syntax (CMS)},
  howpublished = {RFC 5990},
  year         = {2010}
}

@ARTICLE{Zhou,
  author  = {Y. Zhou and D. Feng},
  title   = {Side-channel attacks: Ten years after its publication and the impacts on cryptographic module security testing},
  journal = {Cryptology ePrint Archive},
  year    = {2005}
}

@ARTICLE{ShoupOAEP,
  author  = {V. Shoup},
  title   = {OAEP reconsidered},
  journal = {J. Cryptol.},
  volume  = {15},
  number  = {4},
  pages   = {223--249},
  year    = {2002}
}

@MISC{rfc8017,
  author       = {K. Moriarty and others},
  title        = {PKCS \#1: RSA cryptography specifications version 2.2},
  howpublished = {RFC 8017},
  month        = {Nov.},
  year         = {2016}
}

@MISC{rfc5116,
  author       = {D. McGrew},
  title        = {An interface and algorithms for authenticated encryption},
  howpublished = {RFC 5116},
  month        = {Jan.},
  year         = {2008}
}

@ARTICLE{McGrewViegaGCM,
  author  = {D. McGrew and J. Viega},
  title   = {The Galois/counter mode of operation (GCM)},
  journal = {NIST Modes Oper. Process},
  volume  = {20},
  pages   = {27},
  year    = {2004}
}

@TECHREPORT{nist800-38D,
  author      = {M. J. Dworkin},
  title       = {Recommendation for block cipher modes of operation: Galois/Counter Mode (GCM) and GMAC},
  institution = {NIST Special Publication 800-38D},
  month       = {Nov.},
  year        = {2007}
}

@MISC{rfc5869,
  author       = {H. Krawczyk and P. Eronen},
  title        = {HMAC-based extract-and-expand key derivation function (HKDF)},
  howpublished = {RFC 5869},
  year         = {2010}
}

@MISC{nistSHA-3,
  author       = {NIST},
  title        = {SHA-3 standard: Permutation-based hash and extendable-output functions},
  howpublished = {FIPS PUB 202},
  month        = {Aug.},
  year         = {2015}
}

@INPROCEEDINGS{BonehZhandry,
  author    = {D. Boneh and M. Zhandry},
  title     = {Secure signatures and chosen ciphertext security in a quantum computing world},
  booktitle = {Proc. Annu. Int. Cryptol. Conf. (CRYPTO)},
  year      = {2013},
  pages     = {361--379}
}

@INPROCEEDINGS{Fujisaki,
  author    = {E. Fujisaki and T. Okamoto},
  title     = {Secure integration of asymmetric and symmetric encryption schemes},
  booktitle = {Proc. Annu. Int. Cryptol. Conf.},
  year      = {1999},
  pages     = {537--554}
}

@INPROCEEDINGS{b2,
  author    = {C. Cremers and M. Feltz},
  title     = {Beyond eCK: Perfect forward secrecy under actor compromise and ephemeral-key reveal},
  booktitle = {Proc. Eur. Symp. Res. Comput. Secur.},
  year      = {2012},
  pages     = {734--751}
}

@INPROCEEDINGS{Canetti2001,
  author    = {R. Canetti and H. Krawczyk},
  title     = {Analysis of key-exchange protocols and their use for building secure channels},
  booktitle = {Proc. Int. Conf. Theory Appl. Cryptographic Techn.},
  year      = {2001},
  pages     = {453--474}
}

@INPROCEEDINGS{LaMacchia2007,
  author    = {B. LaMacchia and K. Lauter and A. Mityagin},
  title     = {Stronger security of authenticated key exchange},
  booktitle = {Proc. Int. Conf. Provable Secur.},
  year      = {2007},
  pages     = {1--16}
}

@ARTICLE{ShoupGames,
  author  = {V. Shoup},
  title   = {Sequences of games: a tool for taming complexity in security proofs},
  journal = {Cryptology ePrint Archive},
  year    = {2004}
}

@INPROCEEDINGS{Bader2015,
  author    = {C. Bader and D. Hofheinz and T. Jager and E. Kiltz and Y. Li},
  title     = {Tightly-secure authenticated key exchange},
  booktitle = {Proc. Theory Cryptogr. Conf.},
  year      = {2015},
  pages     = {629--658}
}

@INPROCEEDINGS{Gjosteen2018,
  author    = {K. Gjøsteen and T. Jager},
  title     = {Practical and tightly-secure digital signatures and authenticated key exchange},
  booktitle = {Proc. Annu. Int. Cryptol. Conf.},
  year      = {2018},
  pages     = {95--125}
}

@MISC{nistFIPS204,
  author       = {NIST, Thinh Dang and others},
  title        = {Module-Lattice-Based Digital Signature Standard},
  howpublished = {FIPS PUB 204},
  year         = {2024}
}

@MISC{rfc8446,
  author       = {E.Rescorla},
  title        = {The transport layer security (TLS) protocol version 1.3},
  howpublished = {RFC 8446},
  year         = {2018}
}

@MISC{rfc5280,
  author       = {R.Housley},
  title        = {RFC 8399: Internationalization updates to RFC 5280},
  year         = {2018}
}

@MISC{rfc6928,
  author       = {J.Chu and others},
  title        = {RFC 6928: Increasing TCP's initial window},
  year         = {2013}
}

@MISC{rfc9528,
  author       = {G. Selander and J. Preu{\ss} Mattsson and F. Palombini},
  title        = {Ephemeral Diffie-Hellman Over COSE (EDHOC)},
  howpublished = {RFC 9528},
  year         = {2024},
  doi          = {10.17487/RFC9528}
}

@ARTICLE{AriasJimenez2025,
  author  = {A. Arias-Jimenez and J. Gallego-Madrid and J. Sanchez-Gomez and R. Marin-Perez},
  title   = {Lightweight authenticated key exchange for low-power IoT networks using EDHOC},
  journal = {Internet of Things},
  volume  = {31},
  pages   = {101539},
  year    = {2025},
  doi     = {10.1016/j.iot.2025.101539}
}

@ARTICLE{Zerrouki2023,
  author  = {F. Zerrouki and S. Ouchani and H. Bouarfa},
  title   = {T2S-MAKEP and T2T-MAKEP: A PUF-based Mutual Authentication and Key Exchange Protocol for IoT devices},
  journal = {Internet of Things},
  volume  = {24},
  pages   = {100953},
  year    = {2023},
  doi     = {10.1016/j.iot.2023.100953}
}

@ARTICLE{CruzPiris2025,
  author  = {L. Cruz-Piris and A. Mar{\'i}n-L{\'o}pez and M. {\'A}lvarez-Campana and M. Sanz and J. I. Moreno and D. Arroyo},
  title   = {Measuring the impact of post quantum cryptography in Industrial IoT scenarios},
  journal = {Internet of Things},
  volume  = {34},
  pages   = {101793},
  year    = {2025},
  doi     = {10.1016/j.iot.2025.101793}
}

@ARTICLE{Hanna2025,
  author  = {Y. Hanna and J. Bozhko and S. Tonyali and R. Harrilal-Parchment and M. Cebe and K. Akkaya},
  title   = {A comprehensive and realistic performance evaluation of post-quantum security for consumer IoT devices},
  journal = {Internet of Things},
  volume  = {33},
  pages   = {101650},
  year    = {2025},
  doi     = {10.1016/j.iot.2025.101650}
}

@article{zhandry2015secure,
  title={Secure identity-based encryption in the quantum random oracle model},
  author={Zhandry, Mark},
  journal={International Journal of Quantum Information},
  volume={13},
  number={04},
  pages={1550014},
  year={2015},
  publisher={World Scientific}
}

@MISC{AWSGreengrassFAQ,
  author       = {{Amazon Web Services}},
  title        = {{AWS IoT Greengrass FAQs}},
  howpublished = {Product documentation},
  url          = {https://aws.amazon.com/greengrass/faqs/},
  note         = {Accessed September 7, 2026}
}

@MISC{AWSGreengrassAuthentication,
  author       = {{Amazon Web Services}},
  title        = {Device authentication and authorization for {AWS IoT Greengrass}},
  howpublished = {{AWS IoT Greengrass Version 2 Developer Guide}},
  url          = {https://docs.aws.amazon.com/greengrass/v2/developerguide/device-auth.html},
  note         = {Accessed September 7, 2026}
}

@MISC{InmarsatErgonEnergy,
  author       = {{Inmarsat}},
  title        = {{BGAN M2M} for utilities},
  howpublished = {Ergon Energy industrial case study},
  url          = {https://developer.inmarsat.com/story/bgan-m2m-for-utilities/},
  note         = {Accessed September 7, 2026}
}

@MISC{rfc5280KeyUsage,
  author       = {D. Cooper and S. Santesson and S. Farrell and S. Boeyen and R. Housley and W. Polk},
  title        = {Internet {X.509} Public Key Infrastructure Certificate and Certificate Revocation List ({CRL}) Profile},
  howpublished = {RFC 5280},
  year         = {2008},
  doi          = {10.17487/RFC5280}
}

@ARTICLE{AsyncFedSAR,
  author  = {Yuchao Hou and Dongming Lv and Minghui Liwang and Jie Wang and Zijian Li and Xiaoyu Xia and Di Wu and Youliang Tian and Tony Q. S. Quek and Xianbin Wang},
  title   = {{AsyncFed-SAR}: Towards Communication-Efficient Asynchronous Federated Learning for {SAR} Image Target Recognition},
  journal = {IEEE Transactions on Mobile Computing},
  pages   = {1--18},
  year    = {2026},
  url     = {https://ieeexplore.ieee.org/document/11498445}
}
\endgroup

\clearpage
\begin{figure}[pos=!t]
\footnotesize
\noindent
\begin{minipage}[t]{0.29\columnwidth}
\vspace{0pt}
\centering
\includegraphics[width=\linewidth,height=3.15cm,keepaspectratio]{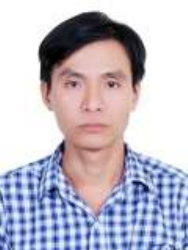}
\end{minipage}\hfill
\begin{minipage}[t]{0.67\columnwidth}
\vspace{0pt}
\textbf{Khiem Pham-Tuan} received the B.S. degree in Informatics Education from Ho Chi Minh City University of Education, Vietnam, in 2006, and the M.Sc. degree in Computer Science from the University of Information Technology, Ho Chi Minh City, Vietnam, in 2016. He is currently a lecturer with the Faculty of Information Technology, Ho Chi Minh City University of Industry and Trade, Vietnam. He is also pursuing the Ph.D. degree in Computer Science at Ho Chi Minh City Open University, Vietnam. His research interests include cryptography, post-quantum cryptography, and computer security.
\end{minipage}

\medskip
\noindent
\begin{minipage}[t]{0.29\columnwidth}
\vspace{0pt}
\centering
\includegraphics[width=\linewidth,height=3.15cm,keepaspectratio]{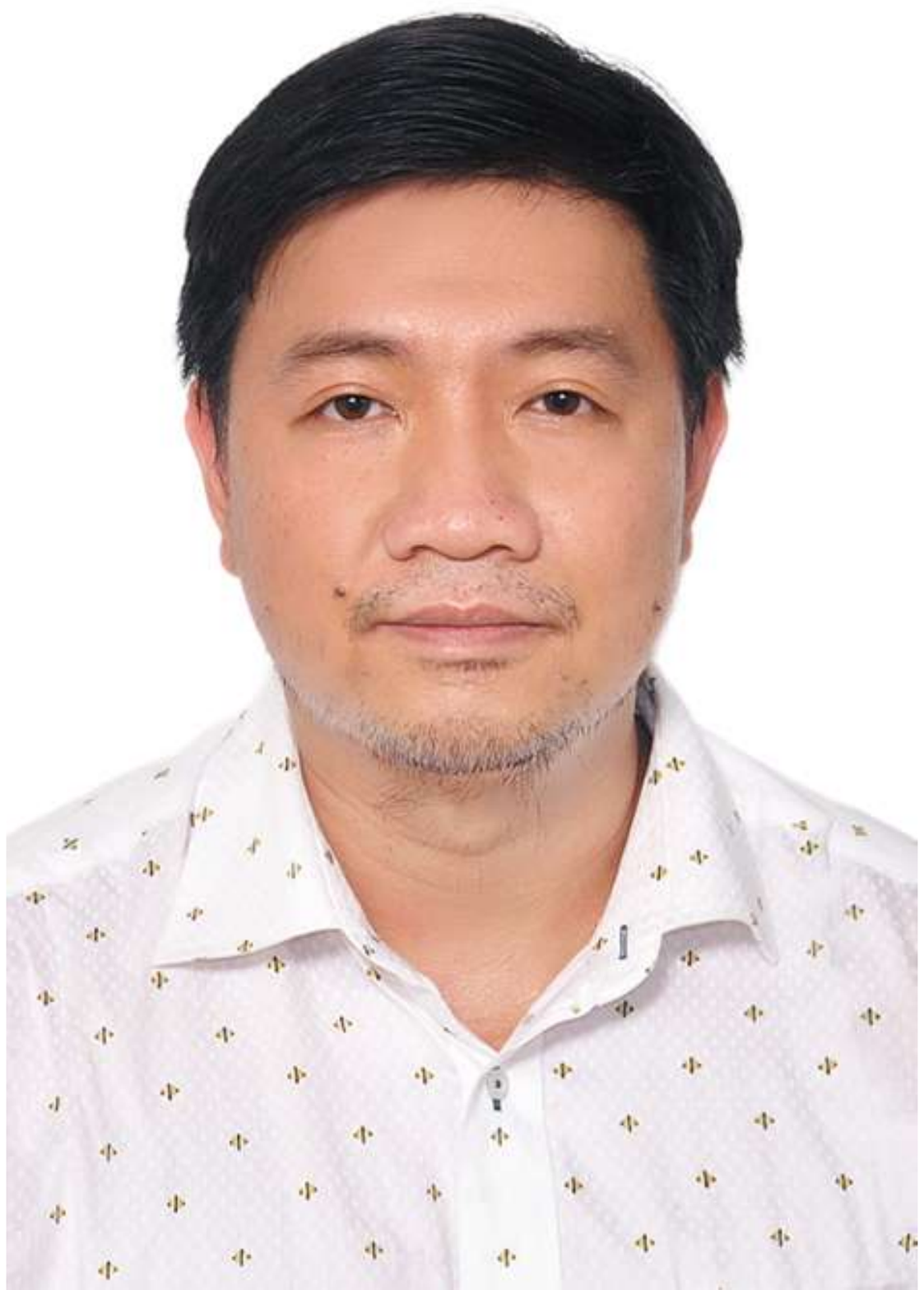}
\end{minipage}\hfill
\begin{minipage}[t]{0.67\columnwidth}
\vspace{0pt}
\textbf{Minh Quang Le} received the Ph.D. degree in Mathematics from the University at Buffalo, USA. He is a lecturer and researcher with the Faculty of Information Technology, Ho Chi Minh City Open University, Vietnam. His research interests lie at the intersection of computational topology, geometry, and deep neural networks. His academic profile is available at \url{https://www.researchgate.net/profile/Minh-Le-115}.
\end{minipage}

\medskip
\noindent
\begin{minipage}[t]{0.29\columnwidth}
\vspace{0pt}
\centering
\includegraphics[width=\linewidth,height=3.15cm,keepaspectratio]{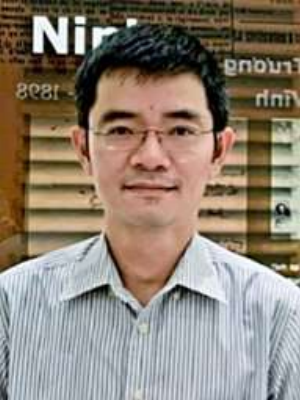}
\end{minipage}\hfill
\begin{minipage}[t]{0.67\columnwidth}
\vspace{0pt}
\textbf{Khuong Nguyen-An} is with the Faculty of Computer Science and Engineering, Ho Chi Minh City University of Technology (HCMUT), Vietnam National University Ho Chi Minh City (VNU-HCM), where he leads the Advanced Cryptography and Cybersecurity Research Group (CYLABK). He received the Ph.D. degree in Mathematics from the University of Groningen, the Netherlands, in 2008. His research interests include applied cryptography, post-quantum cryptography, cybersecurity, blockchain and distributed systems, and trustworthy artificial intelligence. He served on the Editorial Board of the \textit{Bulletin of the Vietnamese Mathematical Society} from 2012 to 2018. His ORCID profile is available at \url{https://orcid.org/0000-0002-9910-6387}.
\end{minipage}
\end{figure}
\clearpage

\end{document}